\documentclass[12pt]{article}

\usepackage{url}
\usepackage{bbm}
\usepackage{mathtools}
\usepackage{amssymb}
\usepackage{amsthm}
\usepackage{empheq}
\usepackage{latexsym}
\usepackage{enumitem}
\usepackage{eurosym}
\usepackage{dsfont}
\usepackage{appendix}
\usepackage{color} 
\usepackage[unicode]{hyperref}
\usepackage{frcursive}
\usepackage[utf8]{inputenc}
\usepackage[T1]{fontenc}
\usepackage{geometry}
\usepackage{multirow}
\usepackage{todonotes}
\usepackage{lmodern}
\usepackage{anyfontsize}
\usepackage{pgfplots}
\usepackage{stmaryrd}
\usepackage{float}

\usepackage{graphicx}
\usepackage{caption}
\usepackage{subcaption}

\graphicspath{ {images/} }

\pgfplotsset{compat=newest}

\definecolor{red}{rgb}{0.7,0.15,0.15}
\definecolor{green}{rgb}{0,0.5,0}
\definecolor{blue}{rgb}{0,0,0.7}
\hypersetup{colorlinks, linkcolor={red},citecolor={green}, urlcolor={blue}}
            
\makeatletter \@addtoreset{equation}{section}

\newtheorem{theorem}{Theorem}
\newtheorem{theorem2}{Theorem}[section]

\newtheorem{lemma}[theorem2]{Lemma}
\newtheorem{proposition}[theorem2]{Proposition}

\newtheorem{definition}[theorem2]{Definition}
\newtheorem{remark}[theorem2]{Remark}

\DeclareUnicodeCharacter{014D}{\=o}
\def \E{\mathbb{E}}

\def \N{\mathbb{N}}

\def \P{\mathbb{P}}

\def \R{\mathbb{R}}

\def\Bc{{\cal B}}

\def\Dc{{\cal D}}

\def\Lc{{\cal L}}

\def\Qc{{\cal Q}}

\usepackage{setspace}
\spacing{1.5}
\def\lf{\left\lfloor}   
\def\rf{\right\rfloor}
\def\lc{\left\lceil}   
\def\rc{\right\rceil}

\title{On bid and ask side-specific tick sizes\footnote{This work benefits from the financial support of the Chaires Analytics and Models for Regulation, Financial Risk. Bastien Baldacci, Joffrey Derchu and Mathieu Rosenbaum gratefully acknowledge the financial support of the ERC Grant 679836 Staqamof. The authors would like to thank Shlomo Ahal for inspiring discussions during the conference "The Regulation and Operation of Modern Financial Markets 2019" in Reykjavik, Iceland. The authors are also grateful to Mouhamad Dramé and Vincent Ragel for insightful comments.}}
\author{Bastien {\sc Baldacci}\footnote{\'Ecole Polytechnique, CMAP, 91128, Palaiseau, France,  bastien.baldacci@polytechnique.edu.} \and Philippe {\sc Bergault}\footnote{Université Paris 1 Panthéon-Sorbonne. Centre d'Economie de la Sorbonne. 106, boulevard de l'Hôpital, 75013 Paris, France, philippe.bergault@etu.univ-paris1.fr} \and Joffrey {\sc Derchu}\footnote{\'Ecole Polytechnique, CMAP, 91128, Palaiseau, France, joffrey.derchu@polytechnique.edu} \and  
Mathieu {\sc Rosenbaum}\footnote{\'Ecole Polytechnique, CMAP, 91128, Palaiseau, France, mathieu.rosenbaum@polytechnique.edu}}

\begin{document}

\maketitle
\begin{abstract}
The tick size, which is the smallest increment between two consecutive prices for a given asset, is a key parameter of market microstructure. In particular, the behavior of high frequency market makers is highly related to its value. We take the point of view of an exchange and investigate the relevance of having different tick sizes on the bid and ask sides of the order book. Using an approach based on the model with uncertainty zones, we show that when side-specific tick sizes are suitably chosen, it enables the exchange to improve the quality of liquidity provision. \\

\noindent{\bf Keywords:} High frequency trading, tick size, model with uncertainty zones, market making, stochastic control, viscosity solutions, financial regulation.
\end{abstract}

\section{Introduction}

The tick size is the smallest increment between two consecutive prices on a trading instrument. It is fixed by the exchange or regulator and typically depends on both the price of the asset and the traded volume, see \cite{huang2016predict,laruelle2018assessing}. It is a crucial parameter of market microstructure and its value is often subject of debates: a too small tick size leads to very frequent price changes whereas a too large tick size prevents the price from moving freely according to the investor's views. In this article, we focus on so-called large tick assets, that is assets for which the spread is most of the time equal to one tick. Such assets represent a large number of financial products, especially in Europe since MIFID II regulation, see \cite{laruelle2018assessing}.  \\

The tick size has a major influence on the ecosystem of financial markets, in particular on the activity of high frequency traders. Being usually considered as market makers, these agents are the main liquidity providers for most heavily traded financial assets. This means that they propose prices at which they are ready to buy (bid price) and sell (ask price) units of financial products. In \cite{frino2015impact}, the authors investigate the behavior of high frequency traders with respect to the relative tick size, which is defined as the ratio between the tick size and the price level. One of their findings is that everything else equal, stocks with a lower relative tick size attract a greater proportion of high frequency traders, see also \cite{dayri2015large,megarbane2017behavior}. This is because they can rapidly marginally adjust their quotes to seize price priority. In the case of a large tick asset, speed is still an important feature as market participants have to compete for queue priority in the order book, see \cite{huang2019glosten,moallemi2016model}. \\

Market makers (typically high frequency traders) face a complex optimization problem: making money out of the bid-ask spread (the difference between the bid and ask prices) while mitigating the inventory risk associated to price changes. This problem is usually addressed via stochastic control theory tools, see for example \cite{avellaneda2008high,cartea2015algorithmic,cartea2014buy,gueant2016financial,gueant2013dealing}. In classical market making models, the so-called efficient price, which represents the market consensus on the value of the asset at a given time, around which the market maker posts his quotes, is a continuous semi-martingale. The quotes of the market maker are continuous in terms of price values and not necessarily multiple of the tick size. However, in actual financial markets, transaction prices are obviously lying on the discrete tick grid. This discreteness of prices is a key feature which cannot be neglected at the high frequency scale since it plays a fundamental role in the design of market making strategies in practice. To get a more realistic market making model, one therefore needs to build a relevant continuous-time price dynamic with discrete state space to take into account this very important microstructural property of the asset. \\

To this end, we borrow the framework of the model with uncertainty zones introduced in \cite{robert2010new,robert2012volatility}. In this model, transaction prices are discrete and the current transaction price is modified only when the underlying continuous efficient price process crosses some predetermined zones. In our approach, we also consider that there exists an efficient price that market participants have in mind when making their trading decisions. Based on this efficient price, market participants build  ``fair'' bid and ``fair'' ask prices. These two prices are lying on the tick grid and represent the views of market participants on reasonable and tradable values for buying and selling, regardless of any inventory constraint. In our setting, depending on his views and his inventory constraint, the market maker chooses whether or not to quote a constant volume at these fair bid and ask prices. This is a stylized viewpoint as in practice the market maker will probably quote a larger spread rather than not quoting at all. The market maker increases (resp. decreases) his current ``fair'' bid price if the efficient price becomes ``sufficiently'' higher (resp. lower) than his current fair bid price and similarly for the ask side. The mechanism to determine whether the efficient price is sufficiently higher (resp. lower) than the current price is that of the model with uncertainty zones, described in Section \ref{section_model_zones}. \\

Usual market making models include a symmetric running penalty for the inventory process, often defined as $\phi\int_0^T Q_t^2 dt$ where $Q_t$ is the inventory of the market maker at time $t\in[0,T]$, $\phi>0$ is a risk aversion parameter and $T$ is the end of the trading period. It is well-known, see for example \cite{adrian2017intraday}, that for regulatory and operational reasons, market participants and especially market makers are reluctant to have a short inventory at the end of the trading day. This is mainly due to constraints imposed by the exchange/regulator and to the overnight repo rate that they have to pay. This asymmetry between long and short terminal inventory of the market maker gives the intuition of the potential relevance of some kind of asymmetry in the market design between buy and sell orders. \\

If some kind of asymmetry is implemented at the microstructure level, it can have important consequences on the profit of exchanges, as it notably depends on the number of processed orders. Typical ways to optimize the number of orders on platforms are the choice of relevant tick sizes and suitable fee schedules (which subsidize liquidity provision and tax liquidity consumption). In \cite{foucault2013liquidity}, the authors highlight the importance of differentiating maker and taker fees in order to increase the trading rate. In the more recent studies \cite{baldacci2019optimal,el2018optimal}, optimal make-take fees schedules are designed based on contract theory. In this work, the asymmetry we consider is not between liquidity consumers and liquidity providers but between buyers and sellers.  \\
 
The goal of this paper is to show the possible benefits for an exchange in terms of liquidity provision of side-specific tick sizes. To this end, we build an agent-based model where a high frequency market maker acts on a large tick asset. The exchange is mitigating the activity on its platform by choosing suitable tick sizes on the bid and ask sides. This means we have a different tick grid for buy and sell orders. For given the tick sizes chosen by the exchange, we formulate the stochastic control problem faced by the market maker who needs to maximize his Profit and Loss (PnL for short) while controlling his inventory risk, taking into account asymmetry between short and long inventory. We show existence and uniqueness of a viscosity solution to the Hamilton-Jacobi-Bellman (HJB for short) equation associated to this problem. Then, we derive a quasi-closed form for the optimal controls of the market maker (up to the value function). In particular, the role of the tick size in the decision of whether or not to quote is explicit: essentially, a large tick size implies a large profit per trade for the market maker but less market orders coming from market takers, and conversely. \\

Next, we solve the optimization problem of the exchange which can select optimal tick sizes knowing the associated trading response of the market maker. In our model, the exchange earns a fixed fee when a transaction occurs. Therefore, its remuneration is related to the quality of the liquidity provided by the market maker on its platform. Numerical results show that side-specific tick sizes are more suitable than symmetric ones both for the market maker and the exchange. The former is able to trigger more alternations in the sign of market orders, which is beneficial both for spread pocketing and inventory management (in contrast with the case where sequences of buy orders are followed by sequences of sell orders). The latter increases the number of transactions on its platform. We also show that a tick size asymmetry can offset short inventory constraints, therefore increasing the gains of both the market maker and the exchange. \\

The paper is organized as follows. In Section \ref{section_model_zones}, we give a reminder on the model with uncertainty zones and explain how we revisit it for market making purposes. The market maker and exchange's problems are described in Section \ref{section_hft_mm_exchange}. We also state here our results about existence and uniqueness of a viscosity solution associated to the control problem of the market maker and derive its optimal controls. Finally, Section \ref{sec_numerical_results} is devoted to numerical results and their interpretations. Proofs are relegated to an appendix.

\section{The model with uncertainty zones}\label{section_model_zones}

In this section, we provide a reminder on the model with uncertainty zones introduced in \cite{robert2010new,robert2012volatility}, and we adapt it to the framework of a market making problem with side-specific tick values. It is commonly admitted that low frequency financial price data behave like a continuous Brownian semi-martingale. However this is clearly not the case for high frequency data. The model with uncertainty zones reproduces sparingly and accurately the behavior of ultra high frequency transaction data of a large tick asset. It is based on a continuous-time semi-martingale efficient price and a one dimensional parameter $\eta\in [0,\frac 1 2]$. The key idea of the model is that when a transaction occurs at some value on the tick grid,  the efficient price is close enough to this value at the transaction time. This proximity is measured through the parameter $\eta$. \\

We define the efficient price $(S_t)_{t\in[0,T]}$ on a filtered probability space $(\Omega,\mathcal{F},\mathbb{P})$ where $T$ is the trading horizon. The logarithm of the efficient price $(Y_t)_{t\in[0,T]}$ is an $\mathcal{F}_t$-adapted continuous Brownian semi-martingale of the form
\begin{align*}
    Y_t=\log(S_t)=\log(S_0)+\int_0^t a_s\mathrm{d}s + \int_0^t \sigma_{s^-}\mathrm{d}W_s,
\end{align*}
where $W$ is an $\mathcal{F}$-Brownian motion, and $(\sigma_t)_{t\in [0,T]}$ is an $\mathcal{F}$-adapted process with càdlàg paths and $(a_t)_{t\in [0,T]}$ is $\mathcal{F}$-progressively measurable. Transaction prices lie on two fixed tick grids, defined by $\{k\alpha^a,k\alpha^b\}$ where $\alpha^a$ (resp. $\alpha^b$) is the tick size on the ask (resp. bid) side and $k\in \mathbb{N}$. For $0\leq \eta^i\leq \frac{1}{2}$ and $i\in\{a,b\}$, we define the zone $U^i_k=[0,\infty)\times (d^i_k,u^i_k)$ with
\begin{align}
    d^i_k=(k+\frac{1}{2}-\eta^i)\alpha^i, u^i_k= (k+\frac{1}{2}+\eta^i)\alpha^i.
\end{align}
Therefore $U_k^a$ is a band of size $2\eta^a\alpha^a$ around the ask mid-tick grid value $(k+\frac{1}{2})\alpha^a$ and $U_k^b$ is a band of size $2\eta^b\alpha^b$ around the bid mid-tick grid value $(k+\frac{1}{2})\alpha^b$. We call these bands the uncertainty zones. The zones on the bid and ask sides are characterized by the parameters $\eta^b,\eta^a$ which control the width of the uncertainty zones. We will see in the next section how the fair bid and ask prices are deduced from the efficient price dynamics across the uncertainty zones. In particular, the larger $\eta^i$, the farther from the last traded price (on the bid or ask side) the efficient price has to be so that a price change occurs. The idea behind the model with uncertainty zones is that, in some sense, market participants feel more comfortable when the asset price is constant than when it is constantly moving. However, there are times when the transaction price has to change because they consider that the last traded price value is not reasonable anymore. \\

For sake of simplicity, we assume that transaction prices cannot jump by more than one tick. We also define the time series of bid and ask transaction times leading to a price change as $(\tau^b_j,\tau_j^a)_{j\geq 0}$. The last traded bid or ask price process is characterized by the couples of transaction times and transaction prices with price changes $(\tau_j; P_{\tau_j^i}^i)_{j\geq 0}$ where $P_{\tau^i_j}^i=S_{\tau^i_j}^{(\alpha^i)}$, the superscript $(\alpha^i)$ denoting the rounding to the nearest $\alpha^i$. \\

The dynamics of the $(\tau_j^i)$ will be described in Section \ref{section_hft_mm_exchange}. One can actually show that the efficient price can be retrieved from transaction data using the equation
\begin{align*}
 S_{\tau_j^i}=S_{\tau_j^i}^{(\alpha^i)}-\alpha^i (\frac{1}{2}-\eta^i)\text{sgn}\big(S_{\tau_j^i}^{(\alpha^i)}-S_{\tau_{j-1}^i}^{(\alpha^i)}\big), \quad i\in\{a,b\}, j\in\mathbb{N}.
\end{align*}
This formula is particularly useful in order to derive ultra high frequency estimators of volatility and covariation (see \cite{robert2012volatility}). The parameters $\eta^i$ can be estimated very easily. Let $N_{\alpha^i,t}^{(a)}$ and $N_{\alpha^i,t}^{(c)}$ be respectively the number of alternations and continuations\footnote{An alternation/continuation corresponds to two consecutive price changes in the opposite/same direction.} of one tick over the period $[0,t]$. Then, an estimator of $\eta^i$ over $[0,t]$ is given by
\begin{align*}
    \hat{\eta}_{\alpha^i,t}=\frac{N_{\alpha^i,t}^{(c)}}{2N_{\alpha^i,t}^{(a)}}.
\end{align*}
We refer to \cite{robert2010new,robert2012volatility} for further details on these estimation procedures. In this paper, we use the model with uncertainty zones for market making purposes rather than for statistical estimation.

\section{High frequency market making under side-specific tick values and interaction with the exchange}\label{section_hft_mm_exchange}

\subsection{The market maker's problem}
We consider a high frequency marker maker acting on an asset whose efficient price $S_t$ has the dynamics
\begin{align*}
    \mathrm{d} S_t=\sigma\mathrm{d}W_t,
\end{align*}
where $\sigma>0$ denotes the volatility of the asset. He uses the model with uncertainty zones described earlier to materialize his views on the fair bid and ask prices. He increases (resp. decreases) his bid price if the efficient price is ``sufficiently'' higher (resp. lower) than his current fair bid price. The notion of ``sufficiently'' higher or lower is determined by the uncertainty zones parameters $\eta^a, \eta^b$, and the tick sizes $\alpha^a,\alpha^b$. If $\eta^a$ is small (resp. large), the market maker changes more (resp. less) frequently his ask price, and similarly for the bid price with $\eta^b$. This leads to the following definition of fair bid and ask prices of the market maker $S^a,S^b$:\footnote{Note that we can have situations where the bid price is above the ask price. However, recall that $S^a$ and $S^b$
are only views about the fair bid and ask prices under the constraint that they have to lie on the tick grids.}
\begin{align*}
   &  S_t^a=S_{t^{-}}^a+\alpha^a \mathbf{1}_{\{ S_t-S^a_{t^{-}}>(\frac{1}{2}+\eta^a)\alpha^a\}}-\alpha^a \mathbf{1}_{\{ S_t-S^a_{t^{-}}<-(\frac{1}{2}+\eta^a)\alpha^a\}}, \\
   &  S_t^b=S_{t^{-}}^b+\alpha^b \mathbf{1}_{\{ S_t-S^b_{t^{-}}>(\frac{1}{2}+\eta^b)\alpha^b\}}-\alpha^b \mathbf{1}_{\{ S_t-S^b_{t^{-}}<-(\frac{1}{2}+\eta^b)\alpha^b\}}.
\end{align*}
Thus the fair bid (resp. ask) is modified when the efficient price is close enough to a new tradable price on the tick grid with mesh $\alpha^b$ (resp. $\alpha^a$).
\begin{remark}
Note that in the case $\alpha^a=\alpha^b,\eta^a=\eta^b$, the fair best bid is equal to the fair best ask. This means that at a given time, a buy or sell order would be at the same price. In this situation, in our stylized view, the market maker would probably quote only on one side (bid or ask). It is consistent with the standard form of the model with uncertainty zones, where, at a given time, transactions can only happen only on one side of the market, depending on the location of the efficient price. Still, the market maker collects the spread from transactions occurring at different times as it is the case in practice.
\end{remark}
We assume a constant volume of transaction equal to one. The market maker can choose to be present or not for a transaction at the bid (with a price $S^b)$ or at the ask (with a price $S^a$). The corresponding cash process at terminal time $T$ is given by
\begin{align*}
    X_T=\int_0^T \Big(  S_t^a\mathrm{d}N_t^a -  S_t^b\mathrm{d}N_t^b\Big),
\end{align*}
where the $N_t^i$ represent the number of transactions on the bid or ask side between $0$ and $t$. In this framework, the inventory of the market maker is given by $Q_t = N_t^b - N_t^a \in \mathcal{Q}=[-\tilde{q},\tilde{q}]$ where $\tilde{q}$ is the risk limit of the market maker. For $i\in\{a,b\}$, the dynamics of $N_t^i$ is that of a point process with intensity 
\begin{align*}
    \lambda(\ell_t^i,Q_t):=\frac{\lambda\ell_t^i}{1+(\kappa\alpha^i)^2}\mathbf{1}_{\{\phi(i)Q_t > -\overline{q}\}}, \quad \phi(i)= \mathbf{1}_{\{i=a\}} - \mathbf{1}_{\{i=b\}}.
\end{align*} 
The process $\ell_t^i\in \{0,1\}$ is the market maker's control which lies in the set of $\mathcal{F}-\text{predictable}$ processes with values in $\{0,1\}$ denoted by $\mathcal{L}$. The parameter $\kappa>0$ controls the sensitivity of the intensities to $\alpha^i$, and $\lambda>0$ is a scale parameter. When the market maker does not want to be present on the bid (resp. ask side) at the price $S^b$ (resp. $S^a)$ he sets $\ell^b=0$ (resp. $\ell^a=0$) and conversely. In our large tick asset setting, the situation where the market maker is not present is a simplified way to model the case where the market maker's quote is higher than the best possible limit. At a given time $t\in[0,T]$, when $\ell_t^b=0$ (resp $\ell_t^a=0)$, the intensity of the point process $N_t^b$ (resp. $N_t^a$) is equal to zero so that there are no incoming transactions. In addition to this, market takers are more confident to send market orders when the tick size is small, as the market maker has more flexibility to adjust his bid and ask prices.\footnote{When the tick size is smaller, the market takers are more willing to trade. This does not necessarily lead to a higher number of orders as it depends on the market maker's presence.} This explains the decreasing shape of the intensities of market order arrivals from market takers with respect to the tick size. The chosen parametric form for the intensities ensures no degenerate behavior when the tick size gets close to zero.  
%\begin{remark}
%In this model, as in most of the literature on optimal market making, the market takers are not strategic and send only aggressive orders with intensities $ \frac{\lambda}{1+(\kappa\alpha^i)^2}$ at the prices offered by the market maker. Nevertheless, transactions can only occur at the fair bid and ask which means that the market takers accept transactions only when they believe that the price is reasonable. 
%\end{remark}

%\begin{remark}
%In this setting, the aggregated demand of the market takers does not equalize the offer. Indeed, such a condition would imply that $ \frac{\lambda}{1+(\kappa\alpha^a)^2}\alpha^a = \frac{\lambda}{1+(\kappa\alpha^b)^2}\alpha^b $, which is not true. Instead, to verify this condition, we should choose intensities $\frac{\lambda^{ab}}{\alpha^a}$ and $\frac{\lambda^{ab}}{\alpha^a}$. However this leads to both the market maker and the exchange to prefer $\alpha^a\to 0$ and $\alpha^b\to 0$, which is not the case in reality. Indeed the intensity at the first limit should be bounded and the trading at the other limits should compensate for the demand imbalance. To circumvent this problem we choose exponentially decreasing intensities at the first limit.
%\end{remark}

The marked-to-market value of the market maker's portfolio at time $t$ is defined as $Q_tS_t$. His optimization problem writes
\begin{align}\label{optimization_pb_MM}
    \sup_{\ell\in\mathcal{L}}\mathbb{E}\Big[X_T +Q_T(S_T-AQ_T)   - \phi\int_0^T Q_s^2 ds - \phi_- \int_0^T |Q_s|^2\mathbf{1}_{Q_s<0}ds\Big],
\end{align}
where $\phi>0$ represents the risk-aversion parameter of the market maker, $\phi_- >0$ is the additional risk aversion of the market maker toward short position on $[0,T]$ and $AQ_T^2$, with $A>0$, is a penalty term for the terminal inventory position regardless of its sign. In this setting, the market maker wishes to hold a terminal inventory close to zero because of the quadratic penalty $AQ_T^2$. The term $\phi\int_0^T Q_s^2 ds$ penalizes long or short positions over the trading period. Problem \eqref{optimization_pb_MM} can of course be rewritten as
\begin{align*}
    \sup_{\ell\in\mathcal{L}}\mathbb{E}\Big[ &Q_T(S_T-AQ_T) +\int_0^T \left\{S_s^a \lambda(\ell^a_s) -  S_s^b\lambda(\ell^b_s) - \phi Q_s^2 - \phi_- Q_s^2\mathbf{1}_{Q_s<0} \right\} ds 
    \Big].
\end{align*}
We define the corresponding value function $h$ defined on the open set
\begin{align*}
\mathcal{D} = \left\{ (S^a,S^b,S) \in \alpha^a \mathbb{Z} \times \alpha^b \mathbb{Z} \times \mathbb{R} \text{ such that } -\left( \frac 12 + \eta^a\right)\alpha^a < S-S^a < \left( \frac 12 + \eta^a\right)\alpha^a\right).\\
\left. \text{ and } -\left( \frac 12 + \eta^b\right)\alpha^b < S-S^b < \left( \frac 12 + \eta^b\right)\alpha^b  \right\}
\end{align*}
by
\begin{align}\label{function_h}
\begin{split}
   h(t,S^a, S^b, S, q) = \sup_{\ell\in\mathcal{L}_t}  \mathbb{E}_{t,S^a, S^b,  S,q}\Big[ &  Q_T(S_T-AQ_T)+\!\int_t^T\!\Big\{S_s^a \lambda(\ell^a_s) -  S_s^b \lambda(\ell^b_s) \\
    & - \phi Q_s^2 - \phi_- Q_s^2\mathbf{1}_{Q_s<0} \Big\}\!\mathrm{d}s
    \Big],  
\end{split}
\end{align}
where $\mathcal{L}_t$ denotes the restriction of admissible controls to $[t,T]$. 
We define the boundary $\partial \mathcal D$ of $\mathcal D$ as 
\begin{align*}
    \partial \mathcal D = \left\{ (S^a,S^b,S) \in \alpha^a \mathbb{Z} \times \alpha^b \mathbb{Z} \times \mathbb{R} \text{ such that } S-S^a = \pm \left( \frac 12 + \eta^a\right)\alpha^a \text{ and/or } S-S^b = \pm \left( \frac 12 + \eta^b\right)\alpha^b . \right\},
\end{align*}
and write $\bar{\mathcal D} = \mathcal D \cup \partial \mathcal D .$ For given $(S^a,S^b)$, if $(S^a,S^b,S)\in \partial \mathcal D$, it means that $S$ corresponds to an efficient price value that triggers a modification of the fair bid or ask price. \\ 

The Hamilton-Jacobi-Bellman equation associated to this stochastic control problem is given by
\begin{align}
     0 = & \partial_t h(t,S^a,S^b,S,q)-\phi q^2 - \phi_- q^2\mathbf{1}_{q<0} + \frac{1}{2}\sigma^2 \partial_{SS}h(t,S^a,S^b,S,q) \nonumber \\
     & + \frac{\lambda}{1+(\kappa\alpha^a)^2}\max_{\ell^a\in \{0,1\}}\Bigg\{ \ell^a\Big(S^a+h(t,S^a,S^b,S,q-\ell^a)-h(t,S^a,S^b,S,q)\Big)\Bigg\} \nonumber\\
     & + \frac{\lambda}{1+(\kappa\alpha^b)^2}\max_{\ell^b\in \{0,1\}}\Bigg\{ \ell^b\Big((-S^b)+h(t,S^a,S^b,S,q+\ell^b)-h(t,S^a,S^b,S,q)\Big)\Bigg\},
\label{HJBMM}
\end{align}
for $(t,S^a,S^b,S,q) \in [0,T) \times \mathcal D \times \mathcal{Q},$ with terminal condition
\begin{align}
    h(T,S^a,S^b,S,q)=q(S-Aq).
    \label{Tcond}
\end{align}

Let us consider the function $h$ defined in \ref{function_h}. For $(t,S^a,S^b,S,q) \in [0,T) \times \mathcal \partial D \times \mathcal{Q}$, and $(t_n,S^a_n,S^b_n,S_n,q_n)_{n\in\N}$ a sequence in $[0,T) \times \mathcal D \times \mathcal{Q}$ which converges to $(t,S^a,S^b,S,q)$, we will show that $h(t_n,S^a_n,S^b_n,S_n,q_n)$ converges independently of the sequence and we denote by $h(t,S^a,S^b,S,q)$ its limit.
On $[0,T) \times \partial \mathcal D \times \mathcal{Q}$, we will show the following boundary conditions (which we will naturally impose for the solution of \ref{HJBMM}):
\begin{align}\label{Bcond}
\begin{split}
      0= &\  \mathbf{1}_{\{S-S^a=(\frac{1}{2}+\eta^a)\alpha^a,\ S-S^b<(\frac{1}{2}+\eta^b)\alpha^b\}}\big(h(t,S^a+\alpha^a,S^b,S,q)-h(t,S^a,S^b,S,q) \big)\\
     & + \mathbf{1}_{\{S-S^a<(\frac{1}{2}+\eta^a)\alpha^a,\ S-S^b=(\frac{1}{2}+\eta^b)\alpha^b\}}\big(h(t,S^a,S^b+\alpha^b,S,q)-h(t,S^a,S^b,S,q) \big)\\
     & + \mathbf{1}_{\{S-S^a=(\frac{1}{2}+\eta^a)\alpha^a,\ S-S^b=(\frac{1}{2}+\eta^b)\alpha^b\}}\big(h(t,S^a+\alpha^a,S^b+\alpha^b,S,q)-h(t,S^a,S^b,S,q) \big) \\
     & + \mathbf{1}_{\{S-S^a=-(\frac{1}{2}+\eta^a)\alpha^a,\ S-S^b>-(\frac{1}{2}+\eta^b)\alpha^b\}}\big(h(t,S^a-\alpha^a,S^b,S,q)-h(t,S^a,S^b,S,q) \big) \\ 
     & + \mathbf{1}_{\{S-S^a>-(\frac{1}{2}+\eta^a)\alpha^a,\ S-S^b=-(\frac{1}{2}+\eta^b)\alpha^b\}}\big(h(t,S^a,S^b-\alpha^b,S,q)-h(t,S^a,S^b,S,q) \big) \\
     & + \mathbf{1}_{\{S-S^a=-(\frac{1}{2}+\eta^a)\alpha^a,\ S-S^b=-(\frac{1}{2}+\eta^b)\alpha^b\}}\big(h(t,S^a-\alpha^a,S^b-\alpha^b,S,q)-h(t,S^a,S^b,S,q) \big).
\end{split}
\end{align}

In other words, the value function varies continuously when the efficient price leaves an uncertainty zone and the prices $S^a$ and $S^b$ are modified.\footnote{Note that, as the terminal condition does not depend on $S^a$ and $S^b$, it also satisfies this boundary condition on $\partial \mathcal D$.} In the following, we say that a function defined on $[0,T) \times \mathcal D \times \mathcal{Q}$ satisfies the continuity conditions if it satisfies \eqref{Bcond}.\\

The following proposition is of particular importance for the existence and uniqueness of a viscosity solution associated to the control problem of the market maker.
\begin{proposition}\label{prop_continuity_verif}
The function $h$ defined in Equation \eqref{function_h} is continuous on $\mathcal{D}$ and satisfies the continuity conditions \eqref{Bcond}.
\end{proposition}
The proof is given in Appendix \ref{proof_continuity_h} and relies on the specific structure of our model based on hitting times of a Brownian motion. We now state the main theorem of this article, whose proof is relegated to Appendix \ref{Proof visco MM}.
\begin{theorem}
The value function $h$ is the unique continuous viscosity solution to Equation \eqref{HJBMM} on $[0,T) \times \mathcal D \times \mathcal{Q}$ with terminal condition \eqref{Tcond} and satisfying the continuity conditions.
\label{viscoMM}
\end{theorem}
The value function depends on five variables. However, as $(S^a,S^b)$ takes value in $\alpha^a \mathbb{N} \times \alpha^b \mathbb{N}$, it can essentially be reduced to three variables as we now explain. For any $(i,j)\in \mathbb{N} ^2 $, we introduce the function $h^{i,j}$ defined on $$[0,T] \times \underbrace{\left(\alpha^a i - (\frac{1}{2}+\eta^a)\alpha^a, \alpha^a i + (\frac{1}{2}+\eta^a)\alpha^a  \right) \cap \left( \alpha^b j -(\frac{1}{2}+\eta^b)\alpha^b, \alpha^b j + (\frac{1}{2}+\eta^b)\alpha^b \right)}_{=\mathcal D _{i,j}} \times \mathcal{Q}$$ by  $h^{i,j}(t,S,q)= h(t,\alpha^a i, \alpha^b j,S,q)$. Then $h^{i,j}$ is the solution of the following HJB equation:
\begin{align*}
     0 = & \partial_t h^{i,j}(t,S,q)-\phi q^2 -  \phi_- (q)_-^2\mathbf{1}_{q<0} + \frac{1}{2}\sigma^2 \partial_{SS}h^{i,j}(t,S,q) \\
     & + \frac{\lambda}{1+(\kappa\alpha^a)^2}\max_{\ell^a\in \{0,1\}}\Bigg\{ \ell^a\Big(\alpha^a i+h^{i,j}(t,S,q-\ell^a)-h^{i,j}(t,S,q)\Big)\Bigg\}\\
     & + \frac{\lambda}{1+(\kappa\alpha^b)^2}\max_{\ell^b\in \{0,1\}}\Bigg\{ \ell^b\Big(-\alpha^b j+h^{i,j}(t,S,q+\ell^b)-h^{i,j}(t,S,q)\Big)\Bigg\},
\end{align*}
with terminal condition $h^{i,j}(T,S,q) = q(S-AQ)$ and natural Dirichlet boundary conditions for $S \in \partial \mathcal D _{i,j}$ :
\begin{align*}
h^{i,j}(t,S,q) = &\ h^{i+1,j}(t,S,q) \mathbf{1}_{\{S-\alpha^a i=(\frac{1}{2}+\eta^a)\alpha^a,\ S-\alpha^b j<(\frac{1}{2}+\eta^b)\alpha^b\}}\\
& + h^{i,j+1}(t,S,q) \mathbf{1}_{\{S-\alpha^a i<(\frac{1}{2}+\eta^a)\alpha^a,\ S-\alpha^b j=(\frac{1}{2}+\eta^b)\alpha^b\}} \\
& + h^{i+1,,j+1}(t,S,q) \mathbf{1}_{\{S-\alpha^a i=(\frac{1}{2}+\eta^a)\alpha^a,\ S-\alpha^b j=(\frac{1}{2}+\eta^b)\alpha^b\}}\\
& + h^{i-1,j}(t,S,q) \mathbf{1}_{\{S-\alpha^a i=-(\frac{1}{2}+\eta^a)\alpha^a,\ S-\alpha^b j>-(\frac{1}{2}+\eta^b)\alpha^b\}}\\
& + h^{i,j-1}(t,S,q) \mathbf{1}_{\{S-\alpha^a i>-(\frac{1}{2}+\eta^a)\alpha^a,\ S-\alpha^b j=-(\frac{1}{2}+\eta^b)\alpha^b\}}\\
& + h^{i-1,j-1}(t,S,q) \mathbf{1}_{\{S-\alpha^a i=-(\frac{1}{2}+\eta^a)\alpha^a,\ S-\alpha^b j=-(\frac{1}{2}+\eta^b)\alpha^b\}}.
\end{align*}
From this, we derive the optimal controls of the market maker as 
\begin{align*}
    & \ell^{\star a}(t,i,j,S,q)=\mathbf{1}_{\{\alpha^a i+h^{i,j}(t,S,q-1)-h^{i,j}(t,S,q)>0 \}},\\
    & \ell^{\star b}(t,i,j,S,q)=\mathbf{1}_{\{-\alpha^b j+h^{i,j}(t,S,q+1)-h^{i,j}(t,S,q)>0 \}}.
\end{align*}

The practical interest of Theorem \ref{viscoMM} is that it allows us to compute the value function and optimal controls based on a finite difference scheme. Examples of computations of the value function are given in Section \ref{sec_numerical_results} and Appendix \ref{app::exemp}.
Having described the problem of the market maker, we now turn to the optimization problem of the platform.

\subsection{The platform's problem}

The market maker acts on a platform whose goal is to maximize the number of market orders on $[0,T]$. The intensities of arrival of market orders are functions of $\ell^a,\ell^b$, which are themselves functions of $\alpha^a,\alpha^b$. We assume that the platform is risk-neutral and earns a fixed taker cost $c>0$ for each market order.\footnote{More complex fee schedules can be handled in this framework. We can for example add a component which is proportional to the amount of cash traded.} Therefore its optimization problem is defined as
\begin{align*}
    \sup_{(\alpha^a,\alpha^b)\in \mathbb{R}_+^2}\mathbb{E}^{l^{\star a},l^{\star b}}\big[X_T^p\big],
\end{align*}
given the optimal controls $(l^{\star a},l^{\star b})$ of the market maker and $X_t^p=c(N_t^a+N_t^b)$. \\

It is easy to observe that this problem boils down to maximizing the function $v$ defined below over $\mathbb{R}^2_+$:
\begin{align*}
    v(\alpha^a, \alpha^b) := \mathbb{E} \left[\int_0^T c\lambda \left\{\frac{\ell^{\star a}(t,S^a_t, S^b_t, S_t, q_t)}{1+(\kappa\alpha^a)^2} + \frac{\ell^{\star b}(t,S^a_t, S^b_t, S_t, q_t)}{1+(\kappa\alpha^b)^2} \right\}\mathrm{d}t \right].
\end{align*}
Here we clearly see the tradeoff of the platform. A small tick size $\alpha^a$ increase the term $(1+(\kappa\alpha^a)^2)^{-1}$. This is because it attracts more buy market orders. However, the optimal control $\ell^{\star, a}$ is more often equal to zero: the gain of the market maker may be too small if he quotes at the price $S^a$, therefore he regularly sets $\ell^{\star,a}=0$. The problem is similar on the bid side. On the other hand, a large tick size increases the gain of the market maker if a transaction occurs, but decreases the number of market orders sent by market takers, hence decreasing the trading volume. \\

We study numerically this problem in the next section by computing the value of $v$ on a two dimensional grid and finding its maximum.

\section{Numerical results}\label{sec_numerical_results}
In this section, we show from numerical experiments the benefits of side-specific tick values in terms of increase of their value function for both the market maker and the platform. Also, we fix reference values $\eta_0$ and $\alpha_0$. From them, to choose the parameter $\eta^i$ associated to a given tick size $\alpha^i$ we use a result from \cite{dayri2015large} which gives the new value of the parameter $\eta^i$ in case of a change of tick size from $\alpha_0$ to $\alpha^i$. This formula writes 
\begin{align}\label{eq::eta}
    \eta^i = \eta_0\sqrt{\frac{\alpha_0}{\alpha^i}}.
\end{align}
In the following, we only consider values of $\alpha^a$ and $\alpha^b$ such that the underlying remains a large tick asset both on the bid and ask sides, that is $\eta^a\leq \frac{1}{2},\eta^b\leq \frac{1}{2}$.  \\

For the first experiments, we set $T=40 s$, $\overline{q}=5$, $\sigma = 0.01 s^{-1}$, $A = 0.1,\kappa=10,\phi=0.005,\lambda=4$, $\eta_0=0.3$ and $\alpha_0=0.01$ which correspond to reasonable values to model a liquid asset. To remain in the large tick regime, we investigate values of $\alpha^i$ satisfying $0.0045\leq \alpha^i\leq 0.05$ for $i=a,b$.  

\subsection{Similar tick values on both sides}

In this section we investigate the case where $\alpha^a=\alpha^b$. We plot in Figure \ref{fig::sym_2} the value functions of the market maker and the exchange, respectively $h$ and $v$, for various values of $\alpha = \alpha^a=\alpha^b$. We fix the efficient price $S=10.5$, the inventory $q=0$ and we only consider values of $\alpha$ so that $0.5/\alpha\in\mathbb{N}$.

\begin{figure}[!h]
\centering
\includegraphics[width=17cm]{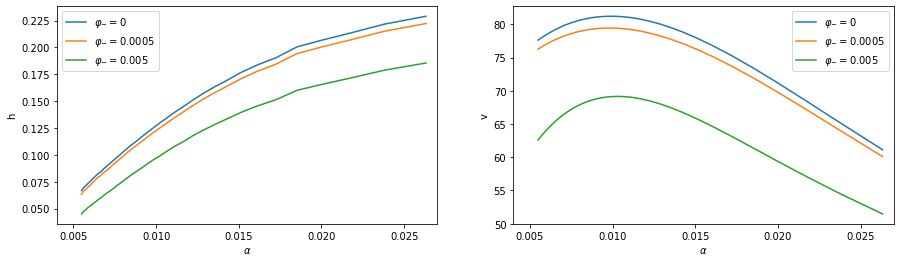}
\caption{Value function $h$ (on the left) and $v$ (on the right) for $\phi_-=0$ in blue, $\phi_-=0.0005$ in orange, $\phi_-=0.005$ in green, as a function of $\alpha = \alpha^a=\alpha^b$.}
\label{fig::sym_2}
\end{figure}

When $\phi_- = 0$, the value of the exchange reaches its maximum at $\alpha\simeq 0.012$. An increase of $\phi_-$ leads to a reduction of the number of transactions. However the optimal tick value for the exchange is not significantly modified. \\

The optimal tick value for the market maker is larger than that of the exchange. This is because the exchange is only interested in attracting orders while the market maker's gain per trade (not taking into account the inventory risk) is linear with respect to the tick value. The trade-off of the exchange is the following: on the one hand, he would like to implement a quite small tick value (to attract market orders) but on the other hand, he must ensure a reasonable presence of market maker.\\

When $\phi_-$ increases, the value function of the market maker decreases, for all tick values. This is no surprise since $\phi_-$ corresponds to an inventory penalization, hence reducing the market maker's PnL. \\

In Figure \ref{fig::sym_diff_h}, we substract the value function when $\phi_-=0$ to the other value functions displayed in Figure \ref{fig::sym_2}. We remark that for the market maker, the larger the tick the more significant the penalization of short inventory in terms of value function. We observe the opposite phenomenon for the exchange: the difference is essentially slightly increasing with respect to $\alpha$. In particular, we see a quite strong impact of the penalization on the value function of the exchange when the tick size is small.

\begin{figure}[!h]
\centering
\includegraphics[width=17cm]{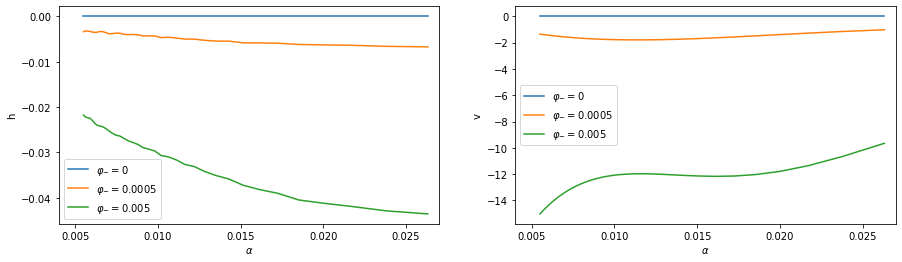}
\caption{Variation of the functions $h$ and $v$ (difference between $\phi_-=0$ in blue, $\phi_-=0.0005$ in orange, $ \phi_-=0.005$ in green, and $\phi_- = 0$ as a function of $\alpha = \alpha^a=\alpha^b$. }
\label{fig::sym_diff_h}
\end{figure}
We now study the case of side-specific tick values.

\subsection{Side-specific tick values: additional opportunities for the market maker}
We set $\alpha^b = 0.0124$ (optimal tick size in the non side-specific case) and let $\alpha^a$ vary. We plot the value functions of the market maker and the exchange in Figure \ref{fig::asym_2}.
\begin{figure}[!h]
\centering
\includegraphics[width=17cm]{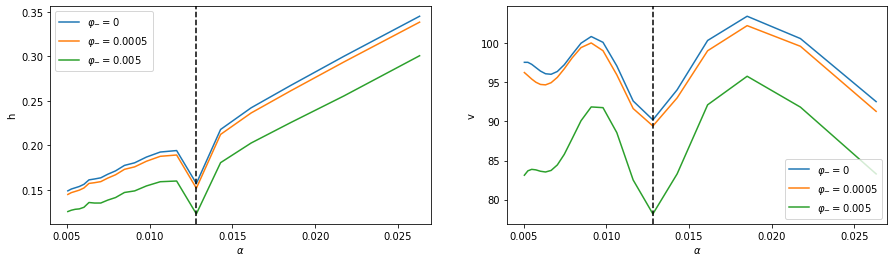}
\caption{Value function $h$ (on the left) and $v$ (on the right) as functions of $\alpha^a$, for $\alpha^b  = 0.0124$, for $\phi_-=0$ in blue, $\phi_-=0.0005$ in orange, $\phi_-=0.005$ in green.}
\label{fig::asym_2}
\end{figure}
Again we observe that both value functions are decreasing with respect to $\phi_-$. From the point of view of the market maker, having non side-specific tick values is sub-optimal, even in the case $\phi_-=0$. This is because when the two tick values are different, it is possible for $S^a$ to be greater than $S^b$ and orders to arrive with the same intensities on both sides: the market maker can collect the spread. It is not possible in the non side-specific case, where the market maker can only pocket the spread from buy and sell orders at two different times. Side-specific tick values are also clearly beneficial for the exchange. The transaction flow increases for $\alpha^a>\alpha^b$ because of the good liquidity provided by the market maker, and for $\alpha^a<\alpha^b$ because of the high number of incoming market orders.

\begin{remark}
Remark that with shifted grids (same tick values on both sides but with a grid shifted compared to the other), those additional opportunities for the market maker would remain. In section \ref{subsec::2d}, we will see however, that, from the point of view of the exchange, side-specific tick values are much more interesting.
\end{remark}

\subsection{Side-specific tick values: effect of $\phi_-$}
\label{subsec::2d}
We plot the two-dimensional value functions of the market maker and the exchange for side-specific tick values.\\

First we take $\phi_-=0$ in Figure \ref{fig::29sym_final}. We note that the opportunity for the market maker mentioned above remains present for all tick values and that the value functions are symmetric around the axis $\alpha^b=\alpha^a$ (side-specific tick values are preferred). Furthermore, we see that the exchange prefers smaller tick values than the market maker. The optimal values for the exchange lie on an anti-diagonal which goes from $(\alpha^a=0.0045,\alpha^b=0.025)$ to $(\alpha^a=0.025,\alpha^b=0.0045)$. On this line the number of transactions varies little. It seems however that the optimum is on the edges of the zone in which the asset remains large tick: the two couples $(\alpha^a,\alpha^b)$ mentioned above.\\

If the tick values are too large the intensities of the market orders become too small and the number of transactions diminishes. If both ticks are too small, the market maker does not trade much because the gain per trade becomes too little compared to the inventory cost (recall that the intensity of market orders is upper bounded). However, the case where one tick is quite small and the other is large is suitable for the market maker: for example, if $\alpha^a<\alpha^b$ his strategy is to be long and liquidate his long position fast if needed thanks to the small value of $\alpha^a$ which ensures a large number of incoming market orders. This explains why the optimal tick values given by the exchange are side-specific and symmetric with respect to the axis $\alpha^a=\alpha^b$. More precisely, the choice of ticks $(\alpha^a=0.0045,\alpha^b=0.025)$ or $(\alpha^a=0.025,\alpha^b=0.0045)$ seems optimal.

\begin{figure}[!h]
\centering
\includegraphics[width=17cm]{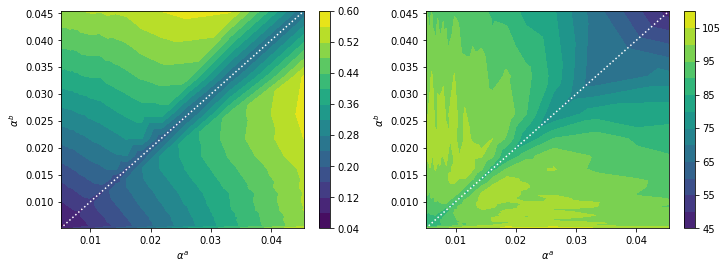}
\caption{Value function $h$ (on the left) and $v$ (on the right) as functions of $\alpha^a$ and $\alpha^b$, for $ \phi_-=0$. }
\label{fig::29sym_final}
\end{figure}

We now plot in Figure \ref{fig::29asym0001_final} the value function for $\phi_- = 0.005$. This non-zero parameter implies a clear decrease of the value function of the market maker and the reduction of the number of transactions. An important remark is that the value functions are no-longer symmetric around the axis $\alpha^b=\alpha^a$.\\

\begin{figure}[!h]
\centering
\includegraphics[width=17cm]{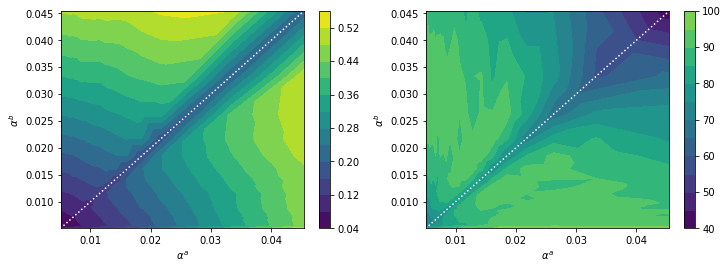}
\caption{Value function $h$ (on the left) and $v$ (on the right) as functions of $\alpha^a$ and $\alpha^b$, for $\phi_-=0.005$. }
\label{fig::29asym0001_final}
\end{figure}

For clarity we plot in Figure \ref{fig::29asym0001_diff_final} the difference of the value functions when $\phi_-=0.005$ and when $\phi_-=0$ as a function of $\alpha^a,\alpha^b$. 
\begin{figure}[!h]
\centering
\includegraphics[width=17cm]{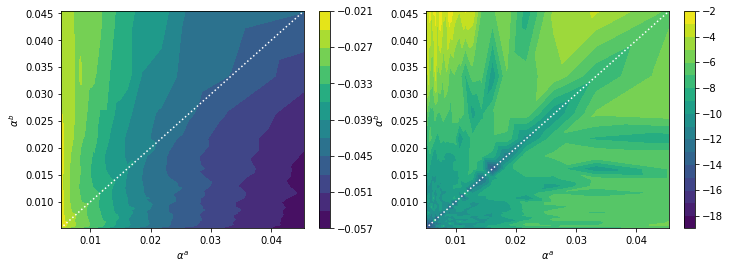}
\caption{Difference between the value function $h$ (on the left) and $v$ (on the right) as functions of $\alpha^a$ and $\alpha^b$, between the case $\phi_-=0.005$ and the case $\phi_-=0$.}
\label{fig::29asym0001_diff_final}
\end{figure}
We see that the added component is not symmetric regarding to the axis $\alpha^b=\alpha^a$ and both the market maker and the exchange tend to prefer the case $\alpha^b>\alpha^a$. It is particularly clear for the market maker's problem where the difference between the values at $(\alpha^a=0.0045,\alpha^b=0.025)$ and $(\alpha^a=0.025,\alpha^b=0.0045)$ is approximately $0.03$ which is roughly $10\%$ of the value function.
Indeed, as explained above, having $\alpha^b$ quite large and $\alpha^a$ rather small essentially ensures that the market maker can maintain a positive inventory all along the trading trajectory: attractive PnL for incoming buy orders and possibility to quickly reduce a positive inventory. \\

The exchange is also more satisfied by the choice $\alpha^b>\alpha^a$. To see that more clearly, we fix $\alpha^a=0.0045$ and plot in Figure \ref{fig::compensation} the value functions $h$ and $v$, as functions of $\alpha^b$, for different values of $\phi$.

\begin{figure}[!h]
\centering
\includegraphics[width=17cm]{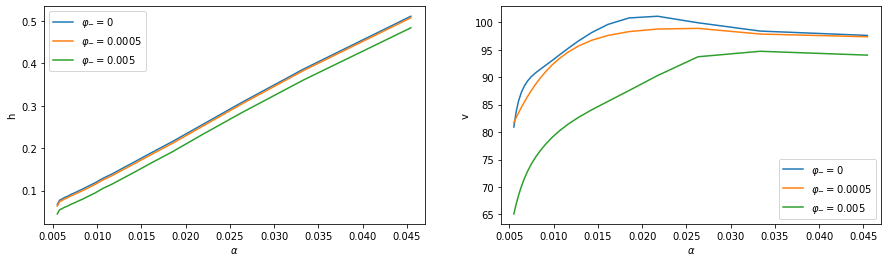}
\caption{Value functions $h$ and $v$ for $\alpha^a=0.0045$, as functions of $\alpha^b$, for different values of $\phi$.}
\label{fig::compensation}
\end{figure}

The value function of the market maker is increasing in $\alpha^b$. This is the same phenomenon as already observed in Figure \ref{fig::29asym0001_final}. The value of the exchange has a maximum which is increasing in $\phi_-$: as the penalization gets more side-specific, the optimal tick values displayed by the exchange become more asymmetric. Indeed, for $\phi_-=0$, the optimum is reached for $\alpha^b\simeq 0.024$, while for $\phi_-=0.005$, it is obtained for $\alpha^b\simeq 0.034$. Note that a relevant tick value set by the exchange can compensate for his loss of value function due to an increase of $\phi_-$. By choosing a new tick size optimally when going from $\phi_-=0$ to $\phi_-=0.005$, the loss in value function is of $7\%$ only. Keeping $\alpha^b= 0.024$ would lead to a loss of $15\%$. Note that the compensation can be total for the market maker (and even exceeds the loss) but is only partial for the exchange. 

\section{Conclusion}
A suitable choice of tick values by the exchange is a subtle equilibrium. If the platform imposes the same tick value on the bid and ask sides, it has to be sufficiently large to ensure significant PnL per trade for the market maker and therefore good liquidity provision, and sufficiently small to attract market orders from market takers. When allowing for side-specific tick values with no constraint on short inventory, the optimal tick values for the exchange are of the form $(\alpha^\star_1,\alpha^\star_2)$ or symmetrically $(\alpha^\star_2,\alpha^\star_1)$ with $\alpha^\star_1<\alpha^\star_2$. In this case, the market maker can take advantage from additional trading opportunities and increase his activity. The exchange benefits from this situation because of the higher number of trades on his platform. Moreover, when there is a penalty for short inventory positions of the market maker, there is only one optimal couple of tick values. In this case, the market maker and subsequently the exchange prefer $\alpha^b>\alpha^a$ and the difference between $\alpha^a$ and $\alpha^b$ at the optimum becomes larger.
Finally, note that side-specific tick values could have subtle consequences in a multi-platform setting. This issue is left for further study, as well as the situation where market takers are more strategic in their execution.

\appendix

\section{Appendix}

\subsection{Proof of Proposition \ref{prop_continuity_verif}}\label{proof_continuity_h}

First we prove the continuity of $h$ on $\Dc\times[0,T)$.\\

Let $q\in\Qc$, $t_1\in[0,T)$, $(s^a, s^b, s_{1})\in\Dc$. Note that $\{s\in\R, (s^a,s^b,s)\in\Dc\}$ is an open interval containing $s_1$, which we denote by $(s^{\leftarrow},s^{\to})$. If the process $S_t$ starts from a point $s\in(s^{\leftarrow},s^{\to})$ with $S^a_t=s^a$ and $S^b_t=s^b$, $S^a_t$ and $S^b_t$ will not jump as long as $S_t$ stays in $(s^{\leftarrow},s^{\to})$. We will prove that the function $(t,s)\in[0,T)\times(s^{\leftarrow},s^{\to})\mapsto h(t,s^a,s^b,s,q)$ is continuous at $(t_1,s_1)$.\\

We fix $\eta>0$. There is a ball with positive diameter $\Bc$ in $[0,T)\times(s^{\leftarrow},s^{\to})$ centered on $(t_1,s_{1})$ and some $\epsilon>0$ such that, if $(t_2,s_{2})\in\Bc$, then
\begin{align}
\label{ineg_list}
    & \E[\tau^1-t_1|S_{t_1}=s_{1}]<\eta, \quad \E[\tau^2-t_2|S_{t_2}=s_{2}]<\eta, \\
    & \P[\tau^1<T|S_{t_1}=s_{1}]>1-\eta, \quad \P[\tau^2<T|S_{t_2}=s_{2}]>1-\eta,
\label{ineg_list2}
\end{align}
and\footnote{These inequalities can be attained independently of the control $\ell$ as $S$ is independent from $Q$.}
\begin{align}
\label{ineg_q}
\begin{split}
    &\underset{\ell\in\Lc}{\inf } \P[\underset{t_1\leq s \leq \tau^1}{\inf } Q_{s} = \underset{t_1\leq s \leq \tau^1}{\sup} Q_{s} = q|S_{t_1}=s_{1},Q_{t_1} = q]>1-\eta, \\ &\underset{\ell\in\Lc}{\inf}\P[\underset{t_2\leq s \leq \tau^2}{\inf} Q_{s} = \underset{t_2\leq s \leq \tau^2}{\sup} Q_{s} = q|S_{t_2}=s_{2},Q_{t_2} = q]>1-\eta,
\end{split}
\end{align}
where 
we write
\begin{align*}
 &\tau^1=T\wedge\inf\{t\geq t_1, S_t^{t_1,s_{1}} = (s_{1}\vee s_{2})+\epsilon \text{ or } S_t^{t_1,s_{1}} = (s_{1}\wedge s_{2})-\epsilon\}, \\
 &\tau^2=T\wedge\inf\{t\geq t_2, S_t^{t_2,s_{2}} = (s_{1}\vee s_{2})+\epsilon \text{ or } S_t^{t_2,s_{2}} = (s_{1}\wedge s_{2})-\epsilon\}.
\end{align*}
The quantities $\tau^1$ and $\tau^2$ are stopping times such that $t_1\leq\tau^1\leq T$ a.s. and $t_2\leq\tau^2\leq T$ a.s.
We impose 
\begin{align*}
s^{\leftarrow}<(s_{1}\wedge s_{2})-\epsilon<s_{1}<(s_{1}\vee s_{2})+\epsilon< s^{\to}, \quad s^{\leftarrow}<(s_{1}\wedge s_{2})-\epsilon<s_{2}<(s_{1}\vee s_{2})+\epsilon< s^{\to}
\end{align*}
for any $(t_2,s_{2})\in\Bc$ by taking a smaller ball $\Bc$ and a smaller $\epsilon$ if necessary. In particular, this tells us that if $(S_{t_1}, S^a_{t_1},S^b_{t_1}) = (s_{1},s^a,s^b)$, $S^a_t$ does not jump between $t_1$ and $\tau^1$. Similarly, if $(S_{t_2}, S^a_{t_2},S^b_{t_2}) = (s_{2},s^a,s^b)$, $S^b_t$ does not jump between between $t_2$ and $\tau^2$ .\\

Let some arbitrary $(t_2,s_{2})\in\Bc$ and $\tau^1$ and $\tau^2$ the associated stopping times. Using the dynamic programming principle, we obtain
\begin{align*}
    h(t_1,s^a, s^b, s_{1},q) = \underset{\ell\in\Lc}{\sup}\E\bigg[ h(\tau^1\!,\!S^a_{\tau^1}\!, \!S^b_{\tau^1}\!, \! S_{\tau^1}\!,\!Q_{\tau^1})\!+\! \int_{t_1}^{\tau^1} \!\!\!\left\{- \phi Q_t^2  - \phi_- Q_t^2\mathbf{1}_{Q_t<0} \right\} dt\Big| S_{t_1} = s_{1}, S^a_{t_1} = s^a, S^b_{t_1} = s^b, Q_{t_1}=q\bigg].
\end{align*}
This can be rewritten as
\begin{align*}
    h(t_1,s^a, s^b, s_{1},q)& = \underset{\ell\in\Lc}{\sup}\E\bigg[\sum\limits_{\bar q\in\Qc}\big( h(\tau^1,s^a, s^b, S_{\tau^1},\bar q)\mathbf{1}_{\{Q_{\tau^1}=\bar q\}}\!+\! \int_{t_1}^{\tau^1}\!\! \left\{- \phi Q_t^2  - \phi_- Q_t^2\mathbf{1}_{Q_t<0} \right\}\!\mathbf{1}_{\{Q_{t}=\bar q\}}\big) dt\big| S_{t_1} \!=\! s_{1}, Q_{t_1}\!=\! q\bigg].
\end{align*}
Recalling that $h$ is bounded, we deduce by \ref{ineg_q} that
\begin{align*}
    \bigg|h(t_1,s^a, s^b, s_{1},q)- \underset{\ell\in\Lc}{\sup}\E\Big[ h(\tau^1,s^a, s^b, S_{\tau^1},q)+ \int_{t_1}^{\tau^1} \left\{- \phi q^2  - \phi_- (q)^2_- \right\} dt\big| S_{t_1} = s_{1}, Q_{t_1}=q\Big]\bigg|\leq C\eta
\end{align*}
for a constant $C$, independent from $(t_1,t_2,s_{1},s_{2})$, and $(q)_-= q^2\mathbf{1}_{q<0}$.
The expectation above does not depend on the control $\ell$, hence we drop the supremum and fix an arbitrary control $\ell=0$. We denote by $\E^0$ the expectation under the probability measure given by this control. The expectation neither depends on the process $Q_t$, so we drop the conditioning with respect to $Q_{t_1}$.\\

This leads to
\begin{align*}
    \bigg|h(t_1,s^a, s^b, s_{1},q)- \E^0\Big[ h(\tau^1,s^a, s^b, S_{\tau^1},q)+ \int_{t_1}^{\tau^1} \left\{- \phi q^2  - \phi_- (q)^2_- \right\} dt\big| S_{t_1} = s_{1}\Big]\bigg|\leq C\eta.
\end{align*}
Similarly, starting from $(t_2,s^a, s^b, s_{2},q)$ with $(t_2,s_2)\in\Bc$, we get
\begin{align*}
    \bigg|h(t_2,s^a, s^b, s_{2},q)- \E^0\Big[ h(\tau^2,s^a, s^b, S_{\tau^2},q)+ \int_{t_2}^{\tau^2} \left\{- \phi q^2  - \phi_- (q)^2_- \right\} dt\big| S_{t_2} = s_{2}\Big]\bigg|\leq C\eta,
\end{align*}
and we deduce that
\begin{align}
\begin{split}
\label{ineq_decomp}
    &|h(t_1,s^a, s^b, s_{1},q)- h(t_2,s^a, s^b, s_{2},q)|\leq \bigg|\E^0\Big[ h(\tau^1,s^a, s^b, S_{\tau^1},q)+ \int_{t_1}^{\tau^1} \left\{- \phi q^2  - \phi_- (q)^2_- \right\} dt| S_{t_1} = s_{1}\Big]\\
   & -\E^0\Big[ h(\tau^2,s^a, s^b, S_{\tau^2},q)+ \int_{t_2}^{\tau^2} \left\{- \phi q^2  - \phi_- (q)^2_- \right\} dt| S_{t_2} = s_{2}\Big]\bigg| +2C\eta\\
    &\leq \bigg|\E^0\Big[ h(\tau^1,s^a, s^b, S_{\tau^1},q)| S_{t_1} = s_{1}\Big]-\E^0\Big[ h(\tau^2,s^a, s^b, S_{\tau^2},q)| S_{t_2} = s_{2}\Big]\bigg|\\
    &+\Big|\phi q^2  - \phi_- (q)^2_-\Big|\Big(\E^0\big[\tau^1-t_1| S_{t_1} = s_{1}\big]+\E^0\big[\tau^2-t_2| S_{t_2} = s_{2}\big]\Big)+2C\eta.
\end{split}
\end{align}

Using \eqref{ineg_list}, we get
\begin{align}
\label{ineq_taus}
    \Big|\E^0\big[\tau^1-t_1| S_{t_1} = s_{1}\big]+E^0\big[\tau^2-t_2| S_{t_2} = s_{2}\big]\Big|<2\eta.
\end{align}

Also, the conditional laws 
\begin{align*}
    & \big(\tau^1| S_{t_1} = s_{1}, S_{\tau^1} = (s_{1}\vee s_{2})+\epsilon, \tau^1<T\big), \quad \big(\tau^1| S_{t_1} = s_{1}, S_{\tau^1} = (s_{1}\wedge s_{2})-\epsilon, \tau^1<T\big),  \\
    & \big(\tau^2| S_{t_2} = s_{2}, S_{\tau^2} = (s_{1}\vee s_{2})+\epsilon, \tau^2<T\big), \quad \big(\tau^2| S_{t_2} = s_{2}, S_{\tau^2} = (s_{1}\wedge s_{2})-\epsilon, \tau^2<T\big),
\end{align*}
have bounded continuous densities, which we denote by $f^{1,+}$, $f^{1,-}$, $f^{2,+}$ and $f^{2,-}$ respectively (see for example \cite{borodinhandbook}, Formula $3.0.6$).
So, by decomposing the first term in \eqref{ineq_decomp} with respect to the values of $S_{\tau^1}$ and $S_{\tau^2}$, we can write
\begin{align}\label{eq_continuity_D_1}
\begin{split}
 &\bigg|\E^0\Big[ h(\tau^1,s^a, s^b, S_{\tau^1},q)| S_{t_1} = s_{1}\Big]-\E^0\Big[ h(\tau^2,s^a, s^b, S_{\tau^2},q)| S_{t_2} = s_{2}\Big]\bigg|\\
    \leq &\Big|\sum\limits_{j\in \{+,-\}}\int_0^T h(t,s^a, s^b, s_j,q)(f^{1,j}(t)\P^0[S_{\tau^1} = s_j, \tau^1<T|S_{t_1} = s_{1}]-f^{2,j}(t)\P^0[S_{\tau^2} = s_j, \tau^2<T|S_{t_2} = s_{2}])dt\Big| \\
    &+\Big|\E^0\big[ h(\tau^1,s^a, s^b, S_{\tau^1},q)\mathbf{1}_{\{S_{\tau^1} \neq s_+, S_{\tau^1} \neq s_-\}\cup\{\tau^1=T\}} | S_{t_1} = s_{1}\big]\Big|\\
    &+\Big|\E^0\big[ h(\tau^2,s^a, s^b, S_{\tau^2},q)\mathbf{1}_{\{S_{\tau^2} \neq s_+, S_{\tau^2} \neq s_-\}\cup\{\tau^2=T\}} | S_{t_2} = s_{2}\big]\Big|   
\end{split}
\end{align}
where $s_+ = s_{1}\vee s_{2}+\epsilon$ and $s_- = s_{1}\wedge s_{2}-\epsilon$. Remark that the event $S_{\tau^1} \neq s_+, S_{\tau^1} \neq s_- , S_{t_1} = s_{1}$ happens only if $\tau^1=T$ so that $\P^0[\{S_{\tau^1} \neq s_+, S_{\tau^1} \neq s_-\}\cup\{\tau^1=T\} | S_{t_1} = s_{1}]<\eta$ by \eqref{ineg_list2}. Similarly $\P^0[\{S_{\tau^2} \neq s_+, S_{\tau^2} \neq s_-\}\cup\{\tau^2=T\} | S_{t_2} = s_{2}]<\eta$ by \eqref{ineg_list2}. As a consequence, using again \eqref{ineq_decomp}, \eqref{ineq_taus} and the fact that $h$ is bounded, we get
\begin{align*}
    &\big|h(t_1,s^a, s^b, s_{1},q)- h(t_2,s^a, s^b, s_{2},q)\big|\\
    & \leq \Big|\sum\limits_{j\in \{+,-\}}\int h(t,s^a, s^b, s_j,q)(f^{1,j}(t)\P^0[S_{\tau^1} = s_j, \tau^1<T|S_{t_1} = s_{1}]-f^{2,j}(t)\P^0[S_{\tau^2} = s_j, \tau^2<T|S_{t_2} = s_{2}])dt\Big|\\
    &+(2|\phi q^2  - \phi_- (q)^2_-|+4C)\eta.
\end{align*}

Recall that the  $f^{1,+}$, $f^{1,-}$, $f^{2,+}$ and $f^{2,-}$ depend on $s_2$ and $t_2$. We have 
\begin{align*}
    &|\P^0[S_{\tau^2} = s_+, \tau^2<T|S_{t_2} = s_{2}]f^{2,+}-\P^0[S_{\tau^1} = s_+, \tau^1<T|S_{t_1} = s_{1}]f^{1,+}| \underset{(t_2,s_{2})\to (t_1,s_{1})}{\to} 0, \\
    &|\P^0[S_{\tau^2} = s_-, \tau^2<T|S_{t_2} = s_{2}]f^{2,-}-\P^0[S_{\tau^1} = s_-, \tau^1<T|S_{t_1} = s_{1}]f^{1,-}| \underset{(t_2,s_{2})\to (t_1,s_{1})}{\to} 0
\end{align*} 
point-wise on $[0,T]$ directly by \cite{borodinhandbook} Formula 3.0.6 and Appendix 11. Having fixed $\epsilon$ and using again \cite{borodinhandbook} Formula 3.0.6 and Appendix 11, we see that the above functions are uniformly bounded with respect to $(s_2, t_2)\in\Bc$. So, using that $h$ is bounded, we can apply the dominated convergence theorem to deduce that 
\begin{align*}
    &\Big|\sum\limits_{j\in \{+,-\}}\int_0^T h(t,s^a, s^b, s_j,q)(f^{1,j}(t)\P^0[S_{\tau^1} = s_j, \tau^1<T|S_{t_1} = s_{1}]-f^{2,j}(t)\P^0[S_{\tau^2} = s_j, \tau^2<T|S_{t_2} = s_{2}])dt\Big|\\
    &\underset{(t_2,s_{2})\to (t_1,s_{1})}{\to} 0.
\end{align*}

Thus we have shown that $h$ is continuous at the point $(t_1,S^a,S^b,s_{1},q)$. The case $t_1=T$ is treated the same way.\\

The continuity conditions can be proved using the same lines: fixing $q\in\Qc$, $t_1\in[0,T)$ and $(S^a,S^b,s_{1})\in\partial\Dc$, choosing $(t_2,s_2)$ close enough to $(t_1,s_1)$ and applying the dynamic programming principle between $t_1$ and $\tau^1$, and $t_2$ and $\tau^2$, for $\tau^1$ and $\tau^2$ two well-chosen stopping times (for example 
\begin{align*}
    \tau^1=T\wedge\inf\big\{t>t_1, S_t=s_{1}+\epsilon\text{ or }S_t = \frac{s^\leftarrow\wedge s_{1}}{2}\big\}, \quad \tau^2=T\wedge\inf\big\{t>t_2, S_t=s_{1}+\epsilon\text{ or }S_t = \frac{s^\leftarrow\wedge s_{1}}{2}\big\}.
\end{align*}
with $\epsilon>0$ small enough, for a boundary inducing an upward jump).

\subsection{Proof of Theorem \ref{viscoMM}}\label{Proof visco MM}

We first prove that the value function of the market maker's problem is indeed a viscosity solution of \eqref{HJBMM}. 

\begin{proposition}
The value function $h$ is a continuous viscosity solution on $[0,T) \times \mathcal D \times \mathcal Q$ of \eqref{HJBMM}. Furthermore, $h(T,S^a, S^b, S, q) = q(S - Aq)$ for all $(S^a, S^b, S, q) \in \mathcal D \times \mathcal Q,$ and 
\begin{align*}
h(t,S^a, S^b, S, q) = &  \mathbf{1}_{\{S-S^a=(\frac{1}{2}+\eta^a)\alpha^a,\ S-S^b<(\frac{1}{2}+\eta^b)\alpha^b\}} h(t,S^a+\alpha^a,S^b,S,q)\\
     {} + & \mathbf{1}_{\{S-S^a<(\frac{1}{2}+\eta^a)\alpha^a,\ S-S^b=(\frac{1}{2}+\eta^b)\alpha^b\}}h(t,S^a,S^b+\alpha^b,S,q)\\
     {} + &\mathbf{1}_{\{S-S^a=(\frac{1}{2}+\eta^a)\alpha^a,\ S-S^b=(\frac{1}{2}+\eta^b)\alpha^b\}}h(t,S^a+\alpha^a,S^b+\alpha^b,S,q)\\
     {} + &\mathbf{1}_{\{S-S^a=-(\frac{1}{2}+\eta^a)\alpha^a,\ S-S^b>-(\frac{1}{2}+\eta^b)\alpha^b\}}h(t,S^a-\alpha^a,S^b,S,q)\\ 
     {} + &\mathbf{1}_{\{S-S^a>-(\frac{1}{2}+\eta^a)\alpha^a,\ S-S^b=-(\frac{1}{2}+\eta^b)\alpha^b\}}h(t,S^a,S^b-\alpha^b,S,q)\nonumber\\
     {} + & \mathbf{1}_{\{S-S^a=-(\frac{1}{2}+\eta^a)\alpha^a,\ S-S^b=-(\frac{1}{2}+\eta^b)\alpha^b\}}h(t,S^a-\alpha^a,S^b-\alpha^b,S,q),
\end{align*}
for all $(t,S^a, S^b, S, q) \in [0,T) \times \partial \mathcal D \times \mathcal Q .$
\end{proposition}

\begin{proof}
Let $(\bar{S}^a, \bar{S}^b, \bar{q}) \in \alpha^a \mathbb{N} \times \alpha^b \mathbb{N}  \times \mathcal Q,$ and $(t_n, S_n)_{n\in \mathbb{N}} \in [0,T] \times \mathbb{R}$ be a sequence such that $$t_n \underset{n\to +\infty}{\to } \hat{t} \in [0,T),$$ $$S_n \underset{n\to +\infty}{\to } \hat{S} \in \mathbb{R},$$ $$h(t_n, \bar{S}^a, \bar{S}^b, S_n, \bar{q}) \underset{n\to +\infty}{\to } h (\hat{t}, \bar{S}^a, \bar{S}^b, \hat{S}, \bar{q}),$$ with $(\bar S ^a, \bar S ^b, \hat S ) \in \mathcal D.$ Without loss of generality we can assume that $(\bar S ^a, \bar S ^b, S_n ) \in \mathcal D$ for all $n\in \mathbb N .$\\

Let us first consider the case $\hat{t} = T$. Let us take two arbitrary controls $\ell^a_s = \ell^b_s = 0$, for all $s \in [0,T),$ then for all $n \in \mathbb{N}$ we have
\begin{align*}
    h(t_n, \bar{S}^a, \bar{S}^b, S_n,\bar{q}) \geq \mathbb{E}_{t_n,\bar{S}^a, \bar{S}^b,  S_n,\bar{q}}\Bigg[  Q_T(S_T-AQ_T)  - \phi \int_{t_n}^T Q_s^2 \mathrm{d}s  - \phi_- \int_{t_n}^T Q_s^2 \mathbf{1}_{Q_s<0} \mathrm{d}s\Bigg],
\end{align*}
and by dominated convergence we can obtain $$h(T,\bar{S}^a, \bar{S}^b, \hat{S}, \bar{q}) \geq \bar{q}(\hat{S} - A\bar{q}).$$

Now let us consider the case $\hat{t}<T.$ Let $\varphi : [0,T) \times \mathcal D \times \mathcal Q \to \mathbb{R}$ be a continuous function, $\mathcal{C}^1$ in $t$, $\mathcal{C}^2$ in $S$ and such that $0 = \underset{[0,T) \times \mathcal{D}}{\min} (h - \varphi) = (h - \varphi)(\hat{t}, \bar{S}^a, \bar{S}^b, \hat{S}, \bar{q}).$ We also assume that $h=\varphi$ only at the point $(\hat{t}, \bar{S}^a, \bar{S}^b, \hat{S}, \bar{q})$. Let us assume that there exists $\eta>0$ such that 
\begin{align*}
     2\eta \leq & \partial_t \varphi(\hat{t},\bar{S}^a,\bar{S}^b,\hat{S},\bar{q})-\phi \bar{q}^2 - \phi_- \bar{q}^2 \mathbf{1}_{\bar{q}<0} + \frac{1}{2}\sigma^2 \partial_{SS}\varphi(\hat{t},\bar{S}^a,\bar{S}^b,\hat{S},\bar{q}) \\
     & +  \frac{\lambda}{1+(\kappa\alpha^a)^2}\max_{\ell^a\in \{0,1\}}\Bigg\{ \ell^a\Big(\bar{S}^a+\varphi(\hat{t},\bar{S}^a,\bar{S}^b,\hat{S},\bar{q}-\ell^a)-\varphi(\hat{t},\bar{S}^a,\bar{S}^b,\hat{S},\bar{q})\Big)\Bigg\}\\
     & + \frac{\lambda}{1+(\kappa\alpha^b)^2}\max_{\ell^b\in \{0,1\}}\Bigg\{ \ell^b\Big(-\bar{S}^b+\varphi(\hat{t},\bar{S}^a,\bar{S}^b,\hat{S},\bar{q}+\ell^b)-\varphi(\hat{t},\bar{S}^a,\bar{S}^b,\hat{S},\bar{q})\Big)\Bigg\}.
\end{align*}
Then we must have 
\begin{align*}
     0 \leq & \partial_t \varphi(t,\bar{S}^a,\bar{S}^b,S,\bar{q})-\phi \bar{q}^2 - \phi_- \bar{q}^2 \mathbf{1}_{\bar{q}<0}  + \frac{1}{2}\sigma^2 \partial_{SS}\varphi(t,\bar{S}^a,\bar{S}^b,S,\bar{q}) \\
     & + \frac{\lambda}{1+(\kappa\alpha^a)^2}\max_{\ell^a\in \{0,1\}}\Bigg\{ \ell^a\Big(\bar{S}^a+\varphi(t,\bar{S}^a,\bar{S}^b,S,\bar{q}-\ell^a)-\varphi(t,\bar{S}^a,\bar{S}^b,S,\bar{q})\Big)\Bigg\}\\
     & + \frac{\lambda}{1+(\kappa\alpha^b)^2}\max_{\ell^b\in \{0,1\}}\Bigg\{ \ell^b\Big(-\bar{S}^b+\varphi(t,\bar{S}^a,\bar{S}^b,S,\bar{q}+\ell^b)-\varphi(t,\bar{S}^a,\bar{S}^b,S,\bar{q})\Big)\Bigg\},
\end{align*}
for all $(t,S) \in B= \left((\hat{t}-r,  \hat{t}+r) \cap [0,T) \right) \times \left(\hat{S}-r, \hat{S}+r\right)$ for some $r >0$. Without loss of generality, we can assume that $B$ contains the sequence $(t_n, S_n)_n$ and that for all $(t,S) \in B,$ we have $(\bar S^a, \bar S^b, S) \in \mathcal D.$ We can choose the value of $\eta$ such that 
$$\varphi(t,\bar{S}^a,\bar{S}^b,S,\bar{q}) + \eta \leq  h(t,\bar{S}^a,\bar{S}^b,S,\bar{q})$$
on $\partial_p B : = \bigg(\left((\hat{t}-r,  \hat{t}+r) \cap [0,T) \right) \times \left(\left\{\hat{S}-r\right\}\cup \left\{\hat{S}+r\right\}\right)\bigg) \cup \bigg( \{\hat{t}+r \} \times \left[\hat{S}-r, \hat{S}+r\right] \bigg).$ We can also assume that
$$\varphi(t,S^a,S^b,S,q) + \eta \leq h(t,S^a,S^b,S,q),$$ 
for $(t, S^a,S^b,S,q) \in \tilde{B}$ with 
\begin{align*}
\tilde{B} = \bigg\{(t,\bar S^a,\bar S^b,S,q) \big| & (t,S) \in B,  q\in \{\bar{q}-1, \bar{q}+1 \}\cap \mathcal Q \bigg\}.  
\end{align*}
We introduce the set 
$$B_{\mathcal{D}} = \left\{(t,\bar{S}^a,\bar{S}^b,S,\bar{q}) \big| (t,S) \in B \right\}$$ 
and set $\pi_n=\inf\{t\geq t_n | (t, S^a_t, S^b_t, S_t, q_t)\notin B_{\mathcal{D}}\}$ with $S^i_{t_n} = \bar{S}^i,$ $q_{t_n} = \bar{q},$ $S_{t_n} = S_n,$ where the processes are controlled by  $$\ell^a_t = \mathbf{1}_{\{ S^a_t + \varphi(t, S^a_t, S^b_t, S_t, q_{t-}-1) - \varphi(t, S^a_t, S^b_t, S_t, q_{t-}) >0\}},$$ $$\ell^b_t = \mathbf{1}_{\{ - S^b_t + \varphi(t, S^a_t, S^b_t, S_t, q_{t-}+1) - \varphi(t, S^a_t, S^b_t, S_t, q_{t-})>0 \}}.$$
{\allowdisplaybreaks
Using Itô's formula and noting that $S_t^a,S_t^b$ do not jump between $t_n$ and $\pi_n$, we derive
\begin{align*}
\varphi(\pi_n, S^a_{\pi_n}, S^b_{\pi_n}, S_{\pi_n},& q_{\pi_n})\!=\!\varphi(t_n, \bar{S}^a, \bar{S}^b, S_n, \bar{q}) + \int_{t_n}^{\pi_n}\!\left\{ \partial_t \varphi(t, S^a_t, S^b_t, S_t, q_t) + \frac 12 \sigma^2 \partial_{S S} \varphi (t, S^a_t, S^b_t, S_t, q_t) \right\}\!dt\\
& \quad + \int_{t_n}^{\pi_n} \lambda(\ell^a_t)\left\{ \varphi (t, S^a_t, S^b_t, S_t, q_{t-} - \ell^a_t) - \varphi (t, S^a_t, S^b_t, S_t, q_{t-})  \right\} dt\\
&\quad + \int_{t_n}^{\pi_n} \lambda(\ell^b_t)\left\{ \varphi (t, S^a_t, S^b_t, S_t, q_{t-} + \ell^b_t) - \varphi (t, S^a_t, S^b_t, S_t, q_{t-})  \right\} dt\\
&\quad + \int_{t_n}^{\pi_n} \sigma \partial_{S} \varphi (t, S^a_t, S^b_t, S_t, q_t) dW_t\\
& \quad + \int_{t_n}^{\pi_n} \left\{ \varphi (t, S^a_t, S^b_t, S_t, q_{t-} - \ell^a_t) - \varphi (t, S^a_t, S^b_t, S_t, q_{t-})  \right\} d\tilde{N}^a_t\\
&\quad + \int_{t_n}^{\pi_n} \left\{ \varphi (t, S^a_t, S^b_t, S_t, q_{t-} + \ell^b_t) - \varphi (t, S^a_t, S^b_t, S_t, q_{t-})  \right\} d\tilde{N}^b_t\\
& \geq \varphi(t_n, \bar{S}^a, \bar{S}^b, S_n, \bar{q})\\
&\quad - \int_{t_n}^{\pi_n} \left\{S_t^a\lambda(\ell^a_t) -  S_t^b \lambda(\ell^b_t) - \phi q_t^2 - \phi_- q_t^2 \mathbf{1}_{q_t<0}\right\} dt\\
&\quad + \int_{t_n}^{\pi_n} \sigma \partial_{S} \varphi (t, S^a_t, S^b_t, S_t, q_t) dW_t\\
& \quad + \int_{t_n}^{\pi_n} \left\{ \varphi (t, S^a_t, S^b_t, S_t, q_{t-} - \ell^a_t) - \varphi (t, S^a_t, S^b_t, S_t, q_{t-})  \right\} d\tilde{N}^a_t\\
&\quad + \int_{t_n}^{\pi_n} \left\{ \varphi (t, S^a_t, S^b_t, S_t, q_{t-} + \ell^b_t) - \varphi (t, S^a_t, S^b_t, S_t, q_{t-})  \right\} d\tilde{N}^b_t.
\end{align*}}
Then by taking the expectation we get
\begin{align*}
    \varphi(t_n, \bar{S}^a, \bar{S}^b, S_n, \bar{q}) \leq & \mathbb{E} \bigg[\varphi(\pi_n, S^a_{\pi_n}, S^b_{\pi_n}, S_{\pi_n}, q_{\pi_n}) + \int_{t_n}^{\pi_n} \left\{S_t^a \lambda(\ell^a_t) -  S_t^b \lambda(\ell^b_t) - \phi q_t^2 - \phi_- q_t^2\mathbf{1}_{q_t<0}\right\} dt \bigg].
\end{align*}
Thus
\begin{align*}
    \varphi(t_n, \bar{S}^a, \bar{S}^b, S_n, \bar{q}) \leq & -\eta +  \mathbb{E} \bigg[h(\pi_n, S^a_{\pi_n}, S^b_{\pi_n}, S_{\pi_n}, q_{\pi_n}) + \int_{t_n}^{\pi_n} \left\{S_t^a \lambda(\ell^a_t) -  S_t^b \lambda(\ell^b_t) - \phi q_t^2 - \phi_- q_t^2\mathbf{1}_{q_t<0} \right\} dt \bigg].
\end{align*}
As 
\begin{align*}
 &\varphi(t_n, \bar{S}^a, \bar{S}^b, S_n, \bar{q}) \underset{n\to +\infty}\to \varphi(\hat{t},\bar{S}^a, \bar{S}^b, \hat{S}, \bar{q}) = h(\hat{t},\bar{S}^a, \bar{S}^b, \hat{S}, \bar{q}), \\
 & h(t_n,  \bar{S}^a, \bar{S}^b, S_n, \bar{q}) \underset{n\to +\infty}\to h(\hat{t},\bar{S}^a, \bar{S}^b, \hat{S}, \bar{q}),
\end{align*}
there exists $n_0\in \mathbb{N}$ such that for all $n\geq n_0$,  $h(t_n,  \bar{S}^a, \bar{S}^b, S_n, \bar{q}) - \frac{\eta}{2} \leq \varphi(t_n, \bar{S}^a, \bar{S}^b, S_n, \bar{q})$ and we deduce
\begin{align*}
    h(t_n, \bar{S}^a, \bar{S}^b, S_n, \bar{q}) \leq & -\frac{\eta}2 +  \mathbb{E} \bigg[h(\pi_n, S^a_{\pi_n}, S^b_{\pi_n}, S_{\pi_n}, q_{\pi_n}) + \int_{t_n}^{\pi_n} \left\{S_t^a \lambda(\ell^a_t) -  S_t^b \lambda(\ell^b_t) - \phi q_t^2  - \phi_- q_t^2 \mathbf{1}_{q_t<0}\right\} dt \bigg],
\end{align*}
which contradicts the dynamic programming principle. Therefore,
\begin{align*}
     0\geq & \partial_t \varphi(\hat{t},\bar{S}^a,\bar{S}^b,\hat{S},\bar{q})-\phi \bar{q}^2 - \phi_- \bar q^2\mathbf{1}_{\bar q<0}+ \frac{1}{2}\sigma^2 \partial_{SS}\varphi(\hat{t},\bar{S}^a,\bar{S}^b,\hat{S},\bar{q}) \\
     & + \frac{\lambda}{1+(\kappa\alpha^a)^2}\max_{\ell^a\in \{0,1\}}\Bigg\{ \ell^a\Big(\bar{S}^a+\varphi(\hat{t},\bar{S}^a,\bar{S}^b,\hat{S},\bar{q}-\ell^a)-\varphi(\hat{t},\bar{S}^a,\bar{S}^b,\hat{S},\bar{q})\Big)\Bigg\}\\
     & + \frac{\lambda}{1+(\kappa\alpha^b)^2}\max_{\ell^b\in \{0,1\}}\Bigg\{ \ell^b\Big(-\bar{S}^b+\varphi(\hat{t},\bar{S}^a,\bar{S}^b,\hat{S},\bar{q}+\ell^b)-\varphi(\hat{t},\bar{S}^a,\bar{S}^b,\hat{S},\bar{q})\Big)\Bigg\},
\end{align*}
and $h$ is a viscosity supersolution of the HJB equation on $[0,T) \times \mathcal{D} \times \mathcal Q.$\\

The proof for the subsolution part is identical.
\end{proof}

For a given $\rho>0$, we introduce the function $\tilde{h}$ such that $$\tilde{h}(t,S^a, S^b, S, q) = e^{\rho t}h(t,S^a, S^b, S, q) \quad \forall \ (t,S^a, S^b, S, q) \in [0,T] \times \mathcal{D} \times \mathcal Q.$$ Then $\tilde{h}$ is a viscosity solution of the following HJB equation:
\begin{align}
0 = & - \rho \tilde{h}(t,S^a,S^b,S,q) + \partial_t \tilde{h}(t,S^a,S^b,S,q)-\phi q^2 - \phi_- q^2\mathbf{1}_{q<0}
+ \frac{1}{2}\sigma^2 \partial_{SS}\tilde{h} \nonumber \\
& + \frac{\lambda}{1+(\kappa\alpha^a)^2}\max_{\ell^a\in \{0,1\}}\Bigg\{ \ell^a\Big(e^{\rho t}S^a+\tilde{h}(t,S^a,S^b,S,q-\ell^a)-\tilde{h}(t,S^a,S^b,S,q)\Big)\Bigg\} \nonumber\\
& + \frac{\lambda}{1+(\kappa\alpha^b)^2}\max_{\ell^b\in \{0,1\}}\Bigg\{ \ell^b\Big(e^{\rho t}(-S^b)+\tilde{h}(t,S^a,S^b,S,q+\ell^b)-\tilde{h}(t,S^a,S^b,S,q)\Big)\Bigg\},
\label{covHJB}     
\end{align}
and we see that proving a maximum principle for \eqref{covHJB} is equivalent to proving one for \eqref{HJBMM}. 

\begin{definition}
Let $U:[0,T) \times \mathcal{D} \times \mathcal{Q} \to \mathbb{R}$ be continuous with respect to $(t,S).$ For $(\hat{t}, \bar{S}^a, \bar{S}^b, \hat{S}, \bar{q}) \in [0,T) \times \mathcal{D} \times \mathcal Q,$ we say that $(y,p,A) \in \mathbb{R}^3$ is in the subjet $\mathcal{P}^- U(\hat{t}, \bar{S}^a, \bar{S}^b, \hat{S}, \bar{q})$ (resp. the superjet $\mathcal{P}^+ U(\hat{t}, \bar{S}^a, \bar{S}^b, \hat{S}, \bar{q})$) if 
\begin{align*}
    & U(t,\bar{S}^a, \bar{S}^b, S, \bar{q}) \geq U(\hat{t}, \bar{S}^a, \bar{S}^b, \hat{S}, \bar{q}) + y(t-\hat{t}) + p(S-\hat{S}) + \frac 12 A(S-\hat{S})^2 + o \left( |t-\hat{t}| + |S - \hat{S}|^2 \right),\\
    &\Big(\text{resp. } U(t,\bar{S}^a, \bar{S}^b, S, \bar{q}) \leq U(\hat{t}, \bar{S}^a, \bar{S}^b, \hat{S}, \bar{q}) + y(t-\hat{t}) + p(S-\hat{S}) + \frac 12 A(S-\hat{S})^2 + o \left( |t-\hat{t}|\!+\!|S - \hat{S}|^2 \right)\Big),
\end{align*}
for all $(t,S)$ such that $(t,\bar{S}^a, \bar{S}^b, S, \bar{q}) \in [0,T) \times \mathcal{D} \times \mathcal Q.$\\

We also define $\bar{\mathcal{P}}^- U(\hat{t}, \bar{S}^a, \bar{S}^b, \hat{S}, \bar{q})$ as the set of points $(y,p,A) \in \mathbb{R}^3$ such that there exists a sequence $(t_n, \bar{S}^a, \bar{S}^b, S_n, \bar{q}, y_n, p_n, A_n) \in [0,T) \times \mathcal{D} \times \mathcal Q \times \mathcal{P}^- U(t_n, \bar{S}^a, \bar{S}^b, S_n, \bar{q})$ satisfying $$(t_n, \bar{S}^a, \bar{S}^b, S_n, \bar{q}, y_n, p_n, A_n) \underset{n \to +\infty}{\to} (\hat{t}, \bar{S}^a, \bar{S}^b, \hat{S}, \bar{q}, y, p, A).$$ The set $\bar{\mathcal{P}}^+ U(\hat{t}, \bar{S}^a, \bar{S}^b, \hat{S}, \bar{q})$ is defined similarly.
\end{definition}

Let us recall one of the definitions of viscosity solutions which we are going to use for the proof of the uniqueness.

\begin{lemma}
A continuous function $\tilde{U}$ is a viscosity supersolution (resp. subsolution) to \eqref{covHJB} on $[0,T) \times \mathcal{D} \times \mathcal Q$ if and only if for all $(\hat{t},\bar{S}^a, \bar{S}^b, \hat{S}, \bar{q}) \in [0,T) \times \mathcal{D}\times \mathcal Q$ and all $(\hat{y}, \hat{p}, \hat{A}) \in \bar{\mathcal{P}}^- U(\hat{t}, \bar{S}^a, \bar{S}^b, \hat{S}, \bar{q})$ (resp. $\bar{\mathcal{P}}^+ U(\hat{t}, \bar{S}^a, \bar{S}^b, \hat{S}, \bar{q})$), we have 
\begin{align*}
 & -\rho \tilde{U}(\hat{t}, \bar{S}^a, \bar{S}^b, \hat{S}, \bar{q}) + \hat{y}-\phi \bar{q}^2 - \phi_- \bar q^2\mathbf{1}_{q<0} + \frac{1}{2}\sigma^2 A\\
 & + \frac{\lambda}{1+(\kappa\alpha^a)^2}\max_{\ell^a\in \{0,1\}}\Bigg\{ \ell^a\Big(e^{\rho t}\bar{S}^a+U(\hat{t},\bar{S}^a,\bar{S}^b,\hat{S},\bar{q}-\ell^a)-U(\hat{t},\bar{S}^a,\bar{S}^b,\hat{S},\bar{q})\Big)\Bigg\}\\
 & + \frac{\lambda}{1+(\kappa\alpha^b)^2}\max_{\ell^b\in \{0,1\}}\Bigg\{ \ell^b\Big(e^{\rho t}(-\bar{S}^b)+U(\hat{t},\bar{S}^a,\bar{S}^b,\hat{S},\bar{q}+\ell^b)-U(\hat{t},\bar{S}^a,\bar{S}^b,\hat{S},\bar{q})\Big)\Bigg\} \leq 0
\end{align*}
(resp. $\geq 0$).
\label{viscolemma}
\end{lemma}

We refer to \cite{bouchard2007introduction} for a proof of this result. We can now state a maximum principle from which the uniqueness can be easily deduced:

\begin{proposition}
Let $U$ (resp. $V$) be a continuous viscosity supersolution (resp. subsolution) of \eqref{HJBMM} with polynomial growth on $[0,T) \times \mathcal{D} \times \mathcal Q$ and satisfying the continuity conditions \eqref{Bcond}. If $U\geq V$ on $\{T\} \times \mathcal{D} \times \mathcal Q$, then $U\geq V$ on $[0,T) \times \mathcal{D} \times \mathcal Q.$
\end{proposition}

\begin{proof}
As before, we introduce the functions $\tilde{U}$ and $\tilde{V}$ such that $$\tilde{U}(t,S^a, S^b, S, q) = e^{\rho t}U(t,S^a, S^b, S, q) \qquad \text{and} \qquad \tilde{V}(t,S^a, S^b, S, q) = e^{\rho t}V(t,S^a, S^b, S, q).$$
Then $\tilde{U}$ and $\tilde{V}$ are respectively viscosity supersolution and subsolution of Equation \eqref{covHJB}
on $[0,T) \times \mathcal{D} \times \mathcal Q$ with $\tilde{U} \geq \tilde{V}$ on $\{T\} \times \mathcal{D} \times \mathcal Q$. To prove the proposition, it is enough to prove that $\tilde{U} \geq \tilde{V}$ on $[0,T) \times \mathcal{D} \times \mathcal Q.$ We proceed by contradiction. Let us assume that $\underset{[0,T) \times \mathcal{D} \times \mathcal Q}{\sup} \tilde{V}-\tilde{U} >0.$ Let $p\in \mathbb{N}^*$ such that $$\underset{\|S\|_2 \to +\infty}{\lim}\quad  \underset{\substack{t \in [0,T], q\in \mathcal Q \\ (S,S^a,S^b)\in\mathcal{D}}}{\sup}\quad  \frac{|\tilde{U}(t, S^a, S^b, S, q)| + |\tilde{V}(t,S^a, S^b, S, q)|}{1+\|S\|_2^{2p}} = 0, $$
where $\|\cdot\|_2$ is the Euclidian norm. Then there exists $(\hat{t}, \bar{S}^a, \bar{S}^b, \hat{S}, \bar{q}) \in [0,T] \times  {\mathcal{D}}\times \mathcal Q$ such that
\begin{align*}
0&< \tilde{V}(\hat{t}, \bar{S}^a, \bar{S}^b, \hat{S}, \bar{q}) - \tilde{U}(\hat{t}, \bar{S}^a, \bar{S}^b, \hat{S}, \bar{q}) - \phi(\hat{t},\hat{S},\hat{S}, \bar{q})\\
&= \underset{(t, S^a, S^b, S, q)}{\sup} \tilde{V}(t, S^a, S^b, S, q) - \tilde{U}(t, S^a, S^b, S, q) - \phi(t,S,S,q),    
\end{align*}
where 
$$\phi(t,S,R,q) := \varepsilon e^{-\mu t} (1 + \|S\|_2^{2p} + \|R\|_2^{2p} ),$$
with $\varepsilon>0,$ $\mu>0$. The choice of the function $\phi$ allows us to look for a supremum in a bounded set with respect to $(S,S^a,S^b)$. Then the supremum is either reached for a point in $[0,T] \times  {\mathcal{D}}\times \mathcal Q$ or on $[0,T] \times\partial\mathcal{D}\times \mathcal Q$ (recall that $\Dc$ is open). But the continuity conditions tell us that if the supremum is reached on $[0,T]\times\partial\mathcal{D}\times\mathcal{Q}$, it is also reached in $[0,T] \times  {\mathcal{D}}\times \mathcal Q$. Since $\tilde{U} \geq \tilde{V}$ on $\{T\} \times \mathcal{D} \times \mathcal Q$, it is clear that $\hat{t}<T$. \\

Then, for all $n\in \mathbb{N}^*$, we can find $(t_n,S_n,R_n) \in [0,T] \times \mathbb{R}^2$ such that $(\bar{S}^a,\bar{S}^b,S_n),(\bar{S}^a,\bar{S}^b,R_n) \in \mathcal D$ and \begin{align*}
0 & <\tilde{V}(t_n,\bar{S}^a,\bar{S}^n, S_n, \bar{q}) - \tilde{U}(t_n,\bar{S}^a,\bar{S}^b, R_n, \bar{q})\\
&\qquad - \phi(t_n,S_n,R_n,\bar q) - n|S_n - R_n|^2 - \left( |t_n-\hat{t}|^2 + |S_n - \hat{S}|^4 \right)\\
& = \underset{(t,S,R) }{\sup} \tilde{V}(t,\bar{S}^a,\bar{S}^b, S, \bar{q}) - \tilde{U}(t,\bar{S}^a,\bar{S}^b, R, \bar{q})\\
&\qquad - \phi(t,S,R,\bar q) - n|S - R|^2 - \left( |t-\hat{t}|^2 + |S - \hat{S}|^4 \right).
\end{align*}

Then, we have
$$(t_n,S_n,R_n) \underset{n\to +\infty}{\to} (\hat{t},\hat{S},\hat{S}),$$
and
\begin{align*}
&\tilde{V}(t_n,\bar{S}^a,\bar{S}^b, S_n, \bar{q}) - \tilde{U}(t_n,\bar{S}^a,\bar{S}^b, R_n, \bar{q})\\
&\qquad - \phi(t_n,S_n,R_n) - n|S_n - R_n|^2 - \left( |t_n-\hat{t}|^2 + |S_n - \hat{S}|^4 \right)\\
& \underset{n\to +\infty}{\to } \tilde{V}(\hat{t}, \bar{S}^a, \bar{S}^b, \hat{S}, \bar{q}) - \tilde{U}(\hat{t}, \bar{S}^a, \bar{S}^b, \hat{S}, \bar{q}) - \phi(\hat{t},\hat{S},\hat{S}).
\end{align*}
For $n\in \mathbb{N}^*$, let us write for $(t,S,R) \in [0,T]\times \mathbb{R}^2$ $$\varphi_n(t,S,R):= \phi(t,S,R) + n|S-R|^2 + |t-\hat{t}|^2 + |S-\hat{S}|^4.$$
Then Ishii's Lemma (see \cite{barles2008second, crandall1992user}) guarantees that for any $\eta>0,$ we can find\\ $(y^1_n,p^1_n,A^1_n) \in \bar{\mathcal{P}}^+ \tilde{V}(t_n,\bar{S}^a, \bar{S}^b, S_n,\bar{q})$ and $(y^2_n,p^2_n,A^2_n) \in \bar{\mathcal{P}}^- \tilde{U}(t_n,\bar{S}^a, \bar{S}^b, R_n,\bar{q})$ such that
$$y^1_n - y^2_n = \partial_t \varphi_n(t_n,S_n,R_n), \quad (p^1_n,p^2_n) = \left(\partial_S \varphi_n, -\partial_R \varphi_n \right)(t_n,S_n,R_n)$$
and
$$\begin{pmatrix} A^1_n & 0\\ 0 & -A^2_n \end{pmatrix} \leq H_{SR} \varphi_n (t_n,S_n,R_n) + \eta \left( H_{SR} \varphi_n (t_n,S_n,R_n) \right)^2, $$
where $H_{SR}\varphi_n (t_n,.,.)$ denotes the Hessian of $\varphi_n(t_n,.,.).$ Applying Lemma \ref{viscolemma}, we get
\begin{align*}
\rho & \left(  \tilde{V}(t_n,\bar{S}^a, \bar{S}^b, S_n, \bar{q}) - \tilde{U}(t_n,\bar{S}^a, \bar{S}^b, R_n, \bar{q}) \right) \leq   y^1_n - y^2_n + \frac 12 \sigma^2 (A^1_n - A^2_n)\\
& + \frac{\lambda}{1+(\kappa\alpha^a)^2}\max_{\ell^a\in \{0,1\}}\Bigg\{ \ell^a\Big(e^{\rho t_n}\bar{S}^a+\tilde{V}(t_n,\bar{S}^a,\bar{S}^b,S_n,\bar{q}-\ell^a)-\tilde{V}(t_n,\bar{S}^a,\bar{S}^b,S_n,\bar{q})\Big)\Bigg\} \nonumber\\
& + \frac{\lambda}{1+(\kappa\alpha^b)^2}\max_{\ell^b\in \{0,1\}}\Bigg\{ \ell^b\Big(e^{\rho t_n}(-\bar{S}^b)+\tilde{V}(t_n,\bar{S}^a,\bar{S}^b,S_n,\bar{q}+\ell^b)-\tilde{V}(t_n,\bar{S}^a,\bar{S}^b,S_n,\bar{q})\Big)\Bigg\}\\
& - \frac{\lambda}{1+(\kappa\alpha^a)^2}\max_{\ell^a\in \{0,1\}}\Bigg\{ \ell^a\Big(e^{\rho t_n}\bar{S}^a+\tilde{U}(t_n,\bar{S}^a,\bar{S}^b,R_n,\bar{q}-\ell^a)-\tilde{U}(t_n,\bar{S}^a,\bar{S}^b,R_n,\bar{q})\Big)\Bigg\} \nonumber\\
& - \frac{\lambda}{1+(\kappa\alpha^b)^2}\max_{\ell^b\in \{0,1\}}\Bigg\{ \ell^b\Big(e^{\rho t_n}(-\bar{S}^b)+\tilde{U}(t_n,\bar{S}^a,\bar{S}^b,R_n,\bar{q}+\ell^b)-\tilde{U}(t_n,\bar{S}^a,\bar{S}^b,R_n,\bar{q})\Big)\Bigg\}.
\end{align*}
Moreover, we have $$ H_{SR} \varphi_n (t_n,S_n,R_n) = \begin{pmatrix} \partial^2_{SS} \phi(t_n,S_n,R_n) + 2n +12 (S_n-\hat{S})^2 & \partial^2_{SR} \phi(t_n,S_n,R_n)-2n\\ \partial^2_{SR} \phi(t_n,S_n,R_n)-2n & \partial^2_{SR}\phi(t_n,S_n,R_n) + 2n  \end{pmatrix}.$$
It follows that
\begin{align*}
\rho & \left(  \tilde{V}(t_n,\bar{S}^a, \bar{S}^b, S_n, \bar{q}) - \tilde{U}(t_n,\bar{S}^a, \bar{S}^b, R_n, \bar{q}) \right) \leq   \partial_t \phi(t_n,S_n,R_n) + 2(t_n-\hat{t})\\
& + \frac 12 \sigma^2 \left(\partial^2_{SS}\phi(t_n,S_n,R_n) + \partial^2_{RR}\phi(t_n,S_n,R_n) + 2\partial^2_{SR}\phi(t_n,S_n,R_n) + 12(S_n-\hat{S})\right) + \eta C_n\\
& + \frac{\lambda}{1+(\kappa\alpha^a)^2}\max_{\ell^a\in \{0,1\}}\Bigg\{ \ell^a\Big(e^{\rho t_n}\bar{S}^a+\tilde{V}(t_n,\bar{S}^a,\bar{S}^b,S_n,\bar{q}-\ell^a)-\tilde{V}(t_n,\bar{S}^a,\bar{S}^b,S_n,\bar{q})\Big)\Bigg\} \nonumber\\
& + \frac{\lambda}{1+(\kappa\alpha^b)^2}\max_{\ell^b\in \{0,1\}}\Bigg\{ \ell^b\Big(e^{\rho t_n}(-\bar{S}^b)+\tilde{V}(t_n,\bar{S}^a,\bar{S}^b,S_n,\bar{q}+\ell^b)-\tilde{V}(t_n,\bar{S}^a,\bar{S}^b,S_n,\bar{q})\Big)\Bigg\}\\
& - \frac{\lambda}{1+(\kappa\alpha^a)^2}\max_{\ell^a\in \{0,1\}}\Bigg\{ \ell^a\Big(e^{\rho t_n}\bar{S}^a+\tilde{U}(t_n,\bar{S}^a,\bar{S}^b,R_n,\bar{q}-\ell^a)-\tilde{U}(t_n,\bar{S}^a,\bar{S}^b,R_n,\bar{q})\Big)\Bigg\} \nonumber\\
& - \frac{\lambda}{1+(\kappa\alpha^b)^2}\max_{\ell^b\in \{0,1\}}\Bigg\{ \ell^b\Big(e^{\rho t_n}(-\bar{S}^b)+\tilde{U}(t_n,\bar{S}^a,\bar{S}^b,R_n,\bar{q}+\ell^b)-\tilde{U}(t_n,\bar{S}^a,\bar{S}^b,R_n,\bar{q})\Big)\Bigg\},
\end{align*}
where $C_n$ does not depend on $\eta.$ Therefore, as the maximums on the right-hand side are always positive, we deduce that for all $n \in \mathbb{N}^*$,
\begin{align*}
\rho & \left(  \tilde{V}(t_n,\bar{S}^a, \bar{S}^b, S_n, \bar{q}) - \tilde{U}(t_n,\bar{S}^a, \bar{S}^b, R_n, \bar{q}) \right) \leq   \partial_t \phi(t_n,S_n,R_n) + 2(t_n-\hat{t})\\
& + \frac 12 \sigma^2 \left(\partial^2_{SS}\phi(t_n,S_n,R_n) + \partial^2_{RR}\phi(t_n,S_n,R_n) + 2\partial^2_{SR}\phi(t_n,S_n,R_n) + 12(S_n-\hat{S})\right)\\
& + \frac{\lambda}{1+(\kappa\alpha^a)^2}\max_{\ell^a\in \{0,1\}}\Bigg\{ \ell^a\Big(e^{\rho t_n}\bar{S}^a+\tilde{V}(t_n,\bar{S}^a,\bar{S}^b,S_n,\bar{q}-\ell^a)-\tilde{V}(t_n,\bar{S}^a,\bar{S}^b,S_n,\bar{q})\Big)\Bigg\} \nonumber\\
& + \frac{\lambda}{1+(\kappa\alpha^b)^2}\max_{\ell^b\in \{0,1\}}\Bigg\{ \ell^b\Big(e^{\rho t_n}(-\bar{S}^b)+\tilde{V}(t_n,\bar{S}^a,\bar{S}^b,S_n,\bar{q}+\ell^b)-\tilde{V}(t_n,\bar{S}^a,\bar{S}^b,S_n,\bar{q})\Big)\Bigg\}.
\end{align*}
As $\tilde{V}$ is continuous and $(t_n,S_n)_n$ converges to $(\hat t, \hat S)$, the last two terms are bounded from above by some constant $M.$ Then by sending $n$ to infinity, we get
\begin{align*}
\rho \left(  \tilde{V}(\hat{t},\bar{S}^a, \bar{S}^b, \hat{S}, \bar{q}) \right. & \left.- \tilde{U}(\hat{t},\bar{S}^a, \bar{S}^b, \hat{S}, \bar{q}) \right) \leq   \partial_t \phi(\hat{t},\hat{S},\hat{S})\\
& + \frac 12 \sigma^2 \left(\partial^2_{SS}\phi(\hat{t},\hat{S},\hat{S}) + \partial^2_{RR}\phi(\hat{t},\hat{S},\hat{S}) + 2\partial^2_{SR}\phi(\hat{t},\hat{S},\hat{S}) \right) + M.
\end{align*}
For $\mu>0$ large enough, the right-hand side is strictly negative, and as $\rho>0$ we get $$\tilde{V}(\hat{t},\bar{S}^a, \bar{S}^b, \hat{S}, \bar{q}) - \tilde{U}(\hat{t},\bar{S}^a, \bar{S}^b, \hat{S}, \bar{q})<0,$$ hence the contradiction. 
\end{proof}
%With the two above propositions, it is easy to conclude the proof of the theorem. Indeed, as $h_*$ is a supersolution such that $h_* \geq h$ on $\{T\} \times \mathcal{D} \times \mathcal Q,$ and $h^*$ is a subsolution such that $h^* \leq h$ on $\{T\} \times \mathcal{D} \times \mathcal Q,$ we can apply the maximum principle to get $h_* \geq h^*$ on $[0,T] \times \mathcal{D} \times \mathcal Q.$  Moreover we also have $h_* \geq h^*$ on $[0,T) \times \partial \mathcal D \times \mathcal Q$, so $h_* \geq h^*$ on $[0,T] \times \bar{\mathcal{D}} \times \mathcal Q.$ But by definition of $h_*$ and $h^*,$ we must have $h_* \leq h \leq h^*$ on $[0,T] \times \bar{\mathcal{D}} \times \mathcal Q,$ which proves that we actually have $h_* = h = h^*$ and $h$ is continuous in $(t,S)$ on $[0,T] \times \bar{\mathcal D } \times \mathcal Q$. Hence we have proved that $h$ is a continuous viscosity solution to Equation \eqref{HJBMM} with terminal condition \eqref{Tcond} and boundary conditions \eqref{Bcond}.\\

%Hence, $h$ is the only continuous viscosity solution to Equation \eqref{HJBMM} with terminal condition \eqref{Tcond} and boundary conditions \eqref{Bcond}.

\subsection{Effects of the uncertainty zones on $h$}
\label{app::exemp}
We keep the same parameters as in Section \ref{sec_numerical_results} and take $\alpha^a=0.01$ and $\alpha^b=0.00625$. We plot the value function of the market maker's problem (the function $h$) on some small range of values of $S$. Note that $S=10.5$ is on both discrete grids.
\begin{figure}[!h]
\centering
\includegraphics[width=18cm]{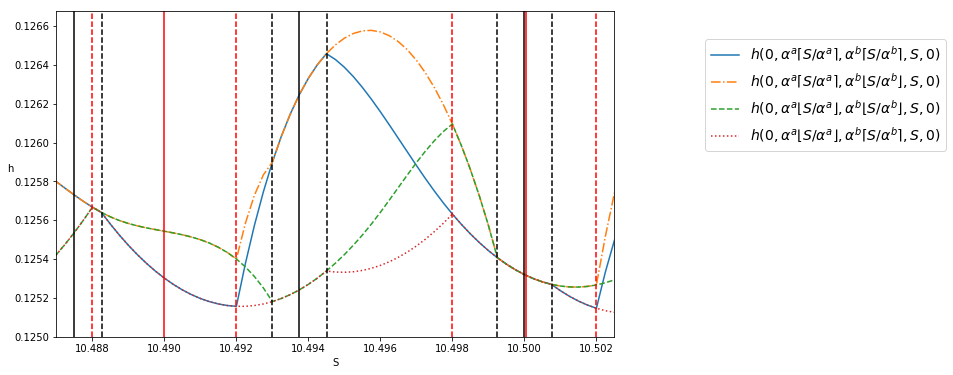}
\caption{Value function $h$ of the market maker for $q=0$, as a function of $S$.}
\label{uz_ef}
\end{figure}
We distinguish 4 possible cases, depending on whether 
\begin{itemize}
    \item $S^a=\alpha^a\lf S/\alpha^a\rf$ and $S^b=\alpha^b\lf S/\alpha^b\rf$ (green dots),
    \item $S^a=\alpha^a\lf S/\alpha^a\rf$ and $S^b=\alpha^b\lc S/\alpha^b\rc$ (red dash-dots),
    \item $S^a=\alpha^a\lc S/\alpha^a\rc$ and $S^b=\alpha^b\lf S/\alpha^b\rf$ (orange dash),
    \item $S^a=\alpha^a\lc S/\alpha^a\rc$ and $S^b=\alpha^b\lc S/\alpha^b\rc$ (blue solid).
\end{itemize}
Note that depending on the value of $S$, some of those cases can be excluded. The solid vertical red and black lines represent respectively the values on the ask ($\alpha^a\mathbb{N}$) and the bid grid ($\alpha^b\mathbb{N}$). The dotted vertical lines represent the limits of the uncertainty zones on each side. \\

In the uncertainty zones, the value function $h$ depends non-trivially on $S^a$ and $S^b$. Thanks to the continuity conditions at the boundaries of the uncertainty zones, we get a smooth behavior of $h$ when $S$ exits a zone. Remark that when $S\in [10 \pm ((\frac{1}{2}-\eta^a)\alpha^a)\wedge((\frac{1}{2}-\eta^b)\alpha^b)]$, necessarily $S^a=S^b=10$. \\

In our example, $\alpha^a>\alpha^b$ and $(\frac{1}{2}-\eta_a)\alpha^a>(\frac{1}{2}-\eta_b)\alpha^b$. So, if $S$ is in  $(10+(\frac{1}{2}-\eta_b)\alpha^b,(\frac{1}{2}-\eta_a)\alpha^a)$, necessarily $S^a=10$, but $S^b$ can take either the value $10$ or $10+\alpha^b$ depending on whether $S$ comes from higher prices or lower prices. This is why there are two curves in the interval $(10+(\frac{1}{2}-\eta_b)\alpha^b,(\frac{1}{2}-\eta_a)\alpha^a)$. At $(\frac{1}{2}-\eta_a)\alpha^a$, two additional curves appear as $S^a$ can also be two different values.

\bibliographystyle{abbrv}
\bibliography{biblio.bib}

\end{document}